\documentclass[journal]{IEEEtran}
\usepackage{amsmath,amssymb,amsfonts}
\usepackage{amsthm}
\usepackage{algorithmic}
\usepackage{array}
\usepackage{textcomp}
\usepackage{stfloats}
\usepackage{url}
\usepackage{verbatim}
\usepackage{graphicx}
\usepackage{cite}
\usepackage{xcolor}
\definecolor{fragfedavg}{RGB}{210,245,215}
\definecolor{fragkrum}{RGB}{224,224,255}
\DeclareRobustCommand{\FragFedAvgTag}{\colorbox{fragfedavg}{\strut\texttt{FragFedAvg}}}
\DeclareRobustCommand{\FragKrumTag}{\colorbox{fragkrum}{\strut\texttt{FragKrum}}}
\newtheorem{theorem}{Theorem}
\newtheorem{corollary}{Corollary}
\newtheorem{lemma}{Lemma}
\newtheorem{assumption}{Assumption}
\theoremstyle{remark}
\newtheorem{remark}{Remark}
\theoremstyle{plain}
\usepackage[ruled,vlined]{algorithm2e}
\usepackage{silence}
\usepackage{caption}
\usepackage{subcaption}

\usepackage{balance}
\usepackage{booktabs}
\usepackage{tabularx}
\usepackage{float}
\usepackage{placeins}
\usepackage{microtype}
\begin{document}

% change symbols of affiliations to numbers
\DeclareRobustCommand*{\IEEEauthorrefmark}[1]{%
  \raisebox{0pt}[0pt][0pt]{\textsuperscript{\footnotesize #1}}%
}

\title{\textit{UnlinkableDFL}: A Framework for Network-Layer Unlinkability in Decentralized Federated Learning}

\author{Chao Feng\IEEEauthorrefmark{1},
Thomas Grübl\IEEEauthorrefmark{1},
Jan von der Assen\IEEEauthorrefmark{1},
Sandrin Raphael Hunkeler\IEEEauthorrefmark{1},
Linn Anna Spitz\IEEEauthorrefmark{1},\\
G\'er\^ome Bovet\IEEEauthorrefmark{2}, and
Burkhard Stiller\IEEEauthorrefmark{1} \\
\IEEEauthorrefmark{1}Communication Systems Group, Department of Informatics, University of Zurich, 8050 Zürich, Switzerland\\
\{cfeng, gruebl, vonderassen, stiller\}@ifi.uzh.ch,
\{sandrinraphael.hunkeler, linnanna.spitz\}@uzh.ch\\
\IEEEauthorrefmark{2}Cyber-Defence Campus, armasuisse Science \& Technology, 3602 Thun, Switzerland\\
gerome.bovet@armasuisse.ch}




\maketitle

\begin{abstract}
Decentralized Federated Learning (DFL) removes the central aggregator of conventional Federated Learning, but peer-to-peer model exchange still exposes network traces: who communicates, when fragments move, and which packets correlate across rounds. This paper studies network-layer sender--message linkability for DFL model sharing and presents \textit{UnlinkableDFL}, a framework in which every participant acts as both a learner and a peer-based mix relay. Shareable model states are split into uniform, onion-encrypted fragment packets and carried over a peer-run mixnet with cover traffic, randomized delays, and independently sampled multi-hop paths. Nodes then perform fragmented aggregation over local and received fragments without sender identities. The analysis bounds sender-linking probability through route uncertainty and relay shuffles, and characterizes when fragment-level aggregation preserves FedAvg-style behavior. A prototype implements QUIC transport, Sphinx-style packets, and Single-Use Reply Block (SURB) acknowledgments. Experiments show that the design sustains learning under sparse deployment while exposing a privacy--cost trade-off: path diversity and relay mixing raise network-layer uncertainty, whereas delay and forwarding dominate overhead. Stress tests confirm robustness to churn and Byzantine updates. A curious-recipient attack marks the boundary of the network-layer guarantee, where payload-level fingerprints survive network-layer anonymization and need complementary defenses, although partial updates and more IID data weaken this attack surface.
\end{abstract}

\begin{IEEEkeywords}
Decentralized Federated Learning, Unlinkability, Anonymous Communication
\end{IEEEkeywords}

% INTRO
\section{Introduction}
Federated Learning (FL) enables multiple data owners to train a shared model without collecting their raw data at a single site~\cite{mcmahan2017communication}. Instead, each participant trains locally and exchanges model updates that are aggregated into improved models. This design reduces direct exposure of training data and helps address privacy and regulatory concerns, but the vanilla Centralized FL (CFL) architecture still depends on a central server to perform model aggregation. The server can become a process bottleneck, a single point of failure, and a powerful observation point for privacy attacks~\cite{Beltran2023}.

Decentralized Federated Learning (DFL) removes this coordinator by letting participants exchange and aggregate updates over a peer-to-peer network~\cite{Beltran2023}. This removes the central aggregation bottleneck and makes the learning process more resilient to server failure. However, this decentralization does not make the exchange anonymous. Although raw data remain local, each training round still leaves observable traces: which peers communicate, when updates are sent, how often nodes appear together, and what statistical features their updates carry. Across rounds, these traces can become linkable. An adversary that correlates them may connect an update or packet to a participant, group fragments that originate from the same node, infer participation patterns, reconstruct private information, or recover the communication topology~\cite{elmrini2024privacy,feng2025topologiestopology}. Thus, DFL shifts privacy risks to the peer-to-peer exchange layer rather than removing them.

This linkability risk spans the data, model, and network layers. Data-layer defenses limit whether a record or peer can be inferred as a training participant~\cite{abadi2016deep}. At the model layer, Shatter uses chunking and virtual identities to reduce content-based attribution, while DivShare supports asynchronous DFL through sliced exchange~\cite{biswas2025noiseless,biswas2025boosting}. These defenses reduce data or payload exposure, not the network trace, which still reveals who communicated, when, and along which paths. Tor- and onion-routing-inspired FL addresses network-layer anonymity~\cite{jadav2023,chen2022,wang2024}, but assumes a central or server-assisted coordinator and does not fit DFL, where peers should provide both learning and anonymous routing. Modern peer-based mixnets can anonymize traffic without reintroducing that coordinator~\cite{piotrowska2017}, but they target independent messages and must be adapted to fragmented, churn-prone, peer-relayed DFL.

This paper, therefore, focuses on network-layer unlinkability for DFL model sharing. The goal is to make model-fragment traffic unlinkable to participant identities at the network layer, in the observable communication trace rather than in decrypted model content, while preserving the decentralized nature of the learning workflow. To this end, we introduce \textit{UnlinkableDFL}, a DFL framework in which each node acts both as a learner and as a mix relay. Shareable model states are split into uniform, onion-encrypted fragment packets and carried over a peer-run mixnet with cover traffic, randomized delays, and independently sampled multi-hop paths. Aggregation is performed locally over a fragment pool containing local and received fragments, so model exchange remains decentralized while the network reduces observable sender--message links. This work makes the following contributions:
\begin{itemize}
    \item Network-layer sender--message unlinkability is formulated as a DFL requirement, with an analysis of how design properties affect network uncertainty.
    \item A peer-based mixnet framework enables each DFL node to perform learning, relaying, mixing, and fragment aggregation without relying on a central entity.
    \item A practical \textit{UnlinkableDFL} prototype combines QUIC-based authenticated transport, fragment-level quantization and compression, padded onion-encrypted packets, randomized delays, cover packets, retransmissions, and fragment-based aggregation.
    \item Experiments quantify learning performance, unlinkability, communication latency, resource cost, and robustness under dynamic peer behavior, including a passive network-layer sender-linking attack that drives linking to the random baseline under the full design.
\end{itemize}

The remainder of the paper covers related work (Section~\ref{sec:related}), the unlinkability problem (Section~\ref{sec:problem}), the framework design (Section~\ref{sec:design}), the prototype (Section~\ref{sec:implementation}), the analysis (Section~\ref{sec:properties}), the evaluation (Section~\ref{sec:evaluation}), and conclusions (Section~\ref{sec:conclusion}).

% RELATED WORK
\section{Related Work}
\label{sec:related}
Linkability captures an adversary's ability to associate observations that were meant to remain separate, such as records, users, updates, fragments, packets, or communication events~\cite{steinbrecher2003}. 

In FL, this risk appears at multiple layers. At the data layer, privacy attacks, such as membership inference attacks (MIAs), test whether a target record or participant contributed to training, often using exposed model parameters or updates. Differentially Private Stochastic Gradient Descent (DP-SGD) mitigates this risk by clipping gradients and adding calibrated noise during training~\cite{abadi2016deep}. This data-layer protection does not prevent an observer from linking transmitted updates, fragments, or packets to the participants that produced them.

At the model layer, linkability concerns updates, chunks, or fragments. An adversary may associate one item with its source or group several as coming from the same participant. Local updates may carry statistical fingerprints of the participant's data distribution, training trajectory, or model state~\cite{elmrini2024privacy}. MixNN mixes neural-network layers from different clients through a proxy before server aggregation to reduce inference leakage from model updates~\cite{lebrun2022mixnn}. Shatter splits model exchanges into chunks and uses virtual identities to reduce content-based attribution in DFL~\cite{biswas2025noiseless}. DivShare slices models for asynchronous DFL, dispersing update content across peers~\cite{biswas2025boosting}. These methods reduce or reshape payload exposure, but the surrounding network trace can still reveal who communicated, when, and along which paths.

At the network layer, linkability concerns communication traces: transmitted updates or fragments, packets, timings, routes, and peer contacts that can identify participants or reveal topology. Mix networks are the classical foundation for breaking these relations, batching and reordering messages across relays to unlink senders from receivers~\cite{chaum1981}. UFL shuffles submitted updates and combines this with Shamir secret sharing~\cite{chen2024}. AIFL provides anonymity for cross-device FL with incentives~\cite{chen2024_02}, and AnoFel supports anonymous update submission for privacy-preserving FL~\cite{almashaqbeh2025}. Onion-routing-inspired systems, including FedOnion, FedTor, and AEFL, hide communication paths or participant identities in FL traffic~\cite{jadav2023,chen2022,wang2024}. These designs reduce network-layer exposure, but remain tied to CFL or server-assisted workflows, where registration, routing, or aggregation is not fully peer-operated. None provides sender--message unlinkability in a setting where peers themselves must supply both learning and anonymous routing.

Messaging mixnets cannot be reused in DFL unchanged. Classical and modern designs, from Chaum's mixnet to Loopix-style stratified mixnets~\cite{piotrowska2017}, assume dedicated relays, stable topologies, and independent fixed-size messages. DFL violates all three assumptions. Every peer is both a learner and a relay, peer sets change under churn, and the payload is a structured model state that must be fragmented and later aggregated by index range rather than delivered as one opaque message. Delivery confirmation adds another obstacle, because a naive acknowledgment can expose the same communication link that the mixnet is meant to hide.

\begin{table}[t]
\centering
\caption{Representative solutions and the linkability problems they address.}
\label{tab:related_comparison}
\footnotesize
\renewcommand{\arraystretch}{1.05}
\setlength{\tabcolsep}{2pt}
\begin{tabular}{@{}>{\raggedright\arraybackslash}p{0.22\columnwidth}
                >{\raggedright\arraybackslash}p{0.45\columnwidth}
                >{\raggedright\arraybackslash}p{0.14\columnwidth}
                >{\raggedright\arraybackslash}p{0.14\columnwidth}@{}}
\hline
\textbf{Solution} & \textbf{Problem addressed} & \textbf{Layer} & \textbf{Setting} \\
\hline
DP-SGD~\cite{abadi2016deep}
& Membership inference
& Data
& CFL/DFL \\
\hline
MixNN~\cite{lebrun2022mixnn}
& Inference leakage from updates
& Model
& CFL \\
% \hline
Shatter~\cite{biswas2025noiseless}
& Fragment-to-source attribution
& Model
& DFL \\
% \hline
DivShare~\cite{biswas2025boosting}
& Fragmented model exchange
& Model
& DFL \\
\hline
AIFL~\cite{chen2024_02}
& Client identity links
& Network
& CFL \\
% \hline
AnoFel~\cite{almashaqbeh2025}
& Sender--message links (updates)
& Network
& CFL \\
% \hline
UFL~\cite{chen2024}
& Sender--message links (updates)
& Network
& CFL \\
% \hline
FedOnion~\cite{jadav2023}
& Path and participant exposure
& Network
& CFL \\
% \hline
FedTor~\cite{chen2022}
& Path and participant exposure
& Network
& CFL \\
% \hline
AEFL~\cite{wang2024}
& Path and participant exposure
& Network
& CFL \\
\hline
\textit{UnlinkableDFL}\newline(\textit{this work})
& Sender--message links (fragments)
& Network
& DFL \\
\hline
\end{tabular}
\end{table}

Table~\ref{tab:related_comparison} summarizes this positioning. Existing defenses address data-layer information leakage, model-layer payload attribution, or network-layer anonymity under CFL assumptions. This paper targets the uncovered case: network-layer sender--message unlinkability for fragmented DFL model sharing without a central shuffler, registry, router, or aggregator.

% Problem
\section{Problem Statement and Threat Model}
\label{sec:problem}
This section defines the network-layer sender--message unlinkability problem and the threat model considered in this work.

\begin{figure*}[!b]
    \centering
    \includegraphics[width=0.8\linewidth]{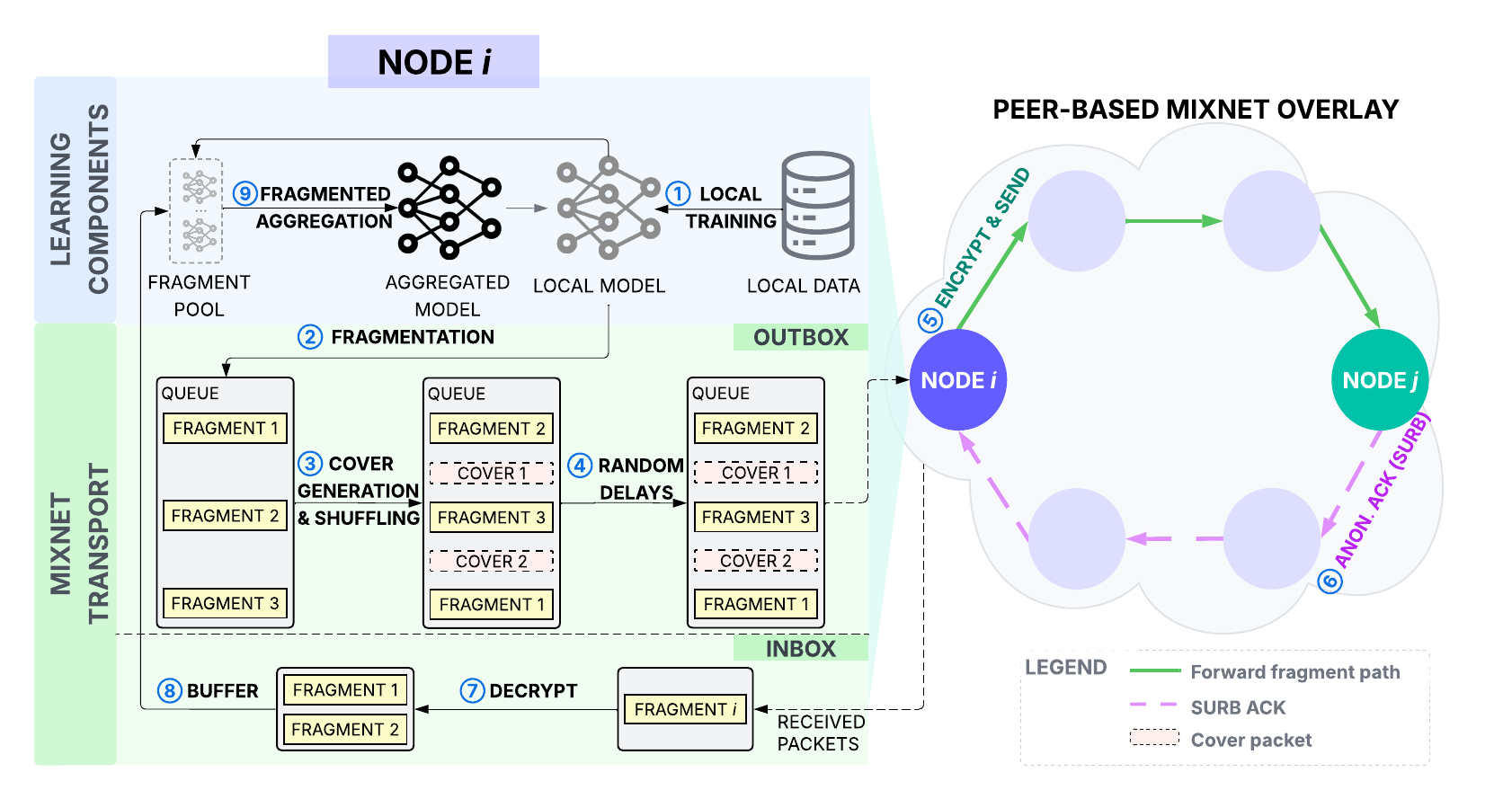}
    \caption{Overview of \textit{UnlinkableDFL}: per-node learning and mixnet transport (left) and the peer-based mixnet overlay (right). Numbered arrows trace one fragment from local training to anonymous aggregation.}
    \label{fig:arch}
\end{figure*}

\subsection{Problem Statement}
\label{sec:problem_statement}
Consider a DFL deployment \(\mathcal{S}\) with topology \(\mathcal{G}=(\mathcal{N},\mathcal{E})\), where \(\mathcal{N}=\{n_1,\ldots,n_N\}\) is the set of nodes. Each node \(n_i\) owns a private dataset \(D_i\), maintains a local model \(w_i^t\) at round \(t\), derives a shareable model state from local training, exchanges model-fragment packets with peers, and aggregates received fragments locally without a global coordinator.

The privacy issue studied here is not whether raw data leave the node, but whether network-layer observations can link a model-sharing message to the node that originated it. In this work, each such message appears on the wire as a model-fragment packet. Let \(\mathsf{Pkt}\) be the set of transmitted model-fragment packets. For each \(m\in\mathsf{Pkt}\), let \(S(m)\in\mathcal{N}\) denote its true sender. Let \(T\) denote the adversary's network-layer observation across rounds, specified in the threat model below.

\paragraph{\textbf{Objective}}
Let \(m_j\) be a target model-fragment packet generated during model sharing, and let \(\mathcal{U}_T\subseteq\mathcal{N}\) be the set of candidate senders that the trace \(T\) leaves consistent with it. After observing \(T\), a passive adversary names the most likely sender, succeeding with probability
\begin{equation}
\label{eq:plink-objective}
    p_{\mathrm{link}}(m_j\mid T)=\max_{u\in\mathcal{U}_T}\Pr[S(m_j)=u\mid T].
\end{equation}
Its advantage over guessing uniformly within \(\mathcal{U}_T\) is
\begin{equation}
\label{eq:adversarial_advantage}
    \mathrm{Adv}^{\mathsf{link}}_{\mathcal{S}}(\mathcal{A})
    = \mathbb{E}_T\!\left[p_{\mathrm{link}}(m_j\mid T)-\frac{1}{|\mathcal{U}_T|}\right],
\end{equation}
where \(\mathcal{S}\) denotes the DFL deployment under analysis. The deployment is \emph{\(\varepsilon\)-unlinkable} if \(\mathrm{Adv}^{\mathsf{link}}_{\mathcal{S}}(\mathcal{A})\le\varepsilon\) for every passive network-layer adversary \(\mathcal{A}\).

\subsection{Threat Model}
\label{sec:threat_model}
\paragraph{\textbf{Adversary Type and Capabilities}}
The adversary \(\mathcal{A}\) is passive and \emph{honest-but-curious}. It may monitor network links and may control a limited subset of compromised nodes \(\mathcal{C}\subset\mathcal{N}\). Compromised nodes follow the deployment rules but share their local observations with \(\mathcal{A}\). In the strongest case, monitored links may cover the whole network.

The adversary observes packet timings, peer contacts, and visible routing metadata across rounds, and each compromised node additionally contributes its own network-layer view---packet timings, transport predecessors, and the next hop it learns by peeling one onion layer. It may know public framework parameters and topology learned through normal participation, but not honest nodes' private datasets, cryptographic keys, or local randomness. It cannot forge, drop, delay, or modify packets, nor break cryptographic primitives. Following the passive global-observer model standard in mixnet analysis, active attacks such as dropping, flooding, and tagging are out of scope. The trace \(T\) is thus network-layer only, and content that a compromised destination obtains by decrypting fragments lies outside \(T\).

\paragraph{\textbf{Adversary Goals and Defensive Objective}}
The adversary aims to link network-layer traffic to participant identities. In particular, it may try to identify the sender of a target packet, decide whether multiple packets originate from the same node, associate fragments with participants, or reconstruct the communication topology \(\mathcal{G}\) from observable traffic patterns. The defensive objective is to keep the resulting sender-linking advantage within the \(\varepsilon\)-unlinkability criterion above while preserving decentralized and adaptive learning under changing peer behavior.

% Framework
\section{UnlinkableDFL Framework}
\label{sec:design}
This section presents \textit{UnlinkableDFL}, a DFL framework for network-layer sender--message unlinkability. It combines fragment generation, payload encoding, mixnet transport, anonymous acknowledgments, and fragmented aggregation.

\subsection{Design Scope and Framework Overview}
\label{subsec:framework_overview}
\textit{UnlinkableDFL} targets decentralized model sharing without runtime central aggregation or routing coordination. Peers act as both learners and relays, selected model fragments are anonymized at the network layer, and learning supports asynchronous execution under topology changes. The learning task and anonymity mechanism are configured separately, so the model, dataset, and aggregation rule can change independently of the mixnet parameters that shape the network trace. Supplementary Material~B shows the configuration interface.

The framework addresses the network-layer linkability defined in Section~\ref{sec:problem}. Figure~\ref{fig:arch} groups each node into learning components and mixnet transport on the left and shows the peer-based mixnet overlay on the right. Steps~1--2 train a local model and convert the resulting shareable model state into indexed fragments. Steps~3--4 add cover packets, shuffle the outbox, and schedule randomized delays before release.

Step~5 onion-encrypts and sends the selected fragment packets over independently sampled multi-hop mixnet paths. Step~6 returns delivery feedback through an attached Single-Use Reply Block (SURB), keeping acknowledgments independent of the receiver's inbox path. Step~7 decrypts incoming packets, and step~8 buffers the recovered fragments before they enter the fragment pool alongside local fragments. Step~9 applies \texttt{FragFedAvg} or \texttt{FragKrum} over available index ranges, replaces the local model with the aggregated result, and repeats this cycle until convergence or a preset round limit.

\subsection{Decentralized Learning Workflow}
\label{subsec:learning_workflow}
A shareable model state is the model representation a node exchanges after local training, such as a parameter update, rather than raw training data. A node trains locally, converts this state into fragments, and submits those fragments to the mixnet transport. The transport does not coordinate learning. It carries uniformly handled packets, and each destination updates its own fragment pool and decides when to aggregate.

A node may produce a shareable model after an epoch, a time interval, or a local trigger, but the output must be representable as indexed fragments that can be routed anonymously and aggregated by index range. Privacy parameters such as \(O\), \(K\), \(\mu\), \(\sigma\), and cover-packet use shape the observable packet trace rather than the learning objective.

\subsection{Model Fragment Generation}
\label{subsec:fragmented_update_generation}
Fragmentation decouples model sharing from sender-level attribution. Instead of exposing one complete model state as a single exchange unit, a node exposes smaller fragments that can be routed independently and mixed with peer fragments. This reduces the content visible in any single packet and avoids a stable one-model-to-one-sender pattern. Figure~\ref{fig:payload-pipeline} summarizes the payload pipeline used after fragmentation. A Sphinx packet~\cite{danezis2009sphinx} here means a fixed-size onion-routed packet with encrypted per-hop routing metadata and a padded payload body.

\begin{figure}[t]
    \centering
    \includegraphics[width=\linewidth]{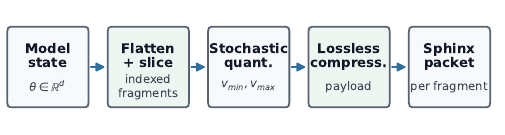}
    \caption{Payload-encoding pipeline.}
    \label{fig:payload-pipeline}
\end{figure}

To generate fragments, a node serializes its shareable model state into a flat parameter vector and partitions the vector into fixed-size slices. Each fragment contains the values of one slice and an index range that identifies its position in the flattened parameter vector. The index is needed for reconstruction and aggregation, but it is not a sender identifier and does not by itself indicate which other ranges are transmitted in the same round.

Before transport, each fragment passes through this payload-encoding pipeline. The default configuration applies unbiased stochastic quantization followed by generic lossless compression, while an unquantized mode bypasses quantization for lossless comparison. Encoding is applied per fragment, which reduces the number of Sphinx packets needed for a model exchange without changing the fragment index semantics.

Fragmentation also supports selective fragment transmission. A node may send only a subset of fragments in a round, for example to reduce bandwidth, react to churn, or support asynchronous training schedules. Receivers insert arriving fragments into the fragment pool, and later aggregation operates over the currently available index ranges rather than waiting for every peer to provide a complete model state.

\subsection{Peer-Based Mixnet Transport}
\label{subsec:communication}
Once a shareable model state has been fragmented and encoded, selected fragments enter the outgoing queue and are drawn into the outbox. The design separates authenticated point-to-point transport from the anonymous mixed overlay: transport links provide reliable neighbor delivery, while unlinkability comes from onion forwarding, randomized delays, cover packets, and independently sampled overlay paths. Unlike a stand-alone messaging mixnet, this transport is embedded in the DFL workflow and must tolerate delayed fragments, selective transmission, and retransmissions.

\paragraph{\textbf{Outbox Mixing and Cover Packets}}
\label{subsec:fragment_outbox}
Each node maintains a local queue for outgoing fragments, relay traffic, cover packets, and SURB-related messages. When the outbox is empty, the node moves up to \(O\) queued items into the outbox. If fewer than \(O\) real items are available, cover packets are added until the outbox reaches size \(O\). The outbox is then randomly permuted using a fresh uniformly sampled permutation. This reshuffling breaks insertion order, local burst boundaries, and the direct timing relationship between local training and outgoing transmissions.

Cover packets prevent idle periods from exposing node activity. They follow the same queueing, shuffling, delay, encryption, and forwarding steps as real fragments. Relays do not distinguish cover packets from model-fragment packets, so packet volume reveals less about whether a node has recently trained, produced, or received fragments.

\paragraph{\textbf{Delay Scheduling and Multi-Hop Routing}}
\label{subsec:onion_routing}
Each outbox item receives an independently sampled delay from a truncated normal distribution with mean \(\mu\) and standard deviation \(\sigma\). Items are transmitted only after their delays expire. Items that remain in the outbox participate in subsequent reshuffling rounds, allowing older and newer packets to mix before release.

Before transmission, the sender samples a multi-hop route from its peer view. The path length is bounded by \(K\), and fragments choose routes independently. Onion encryption makes each relay learn only the next hop. A relay peels one layer, requeues the packet, and releases it after outbox mixing rather than forwarding it immediately. This reduces route-based sender--message cues and prevents same-model fragments from following a stable path pattern.

\paragraph{\textbf{Incoming Packet Processing}}
Incoming packets first enter the inbox. Relays peel one onion layer and requeue the packet for later mixing, while destinations decrypt the packet, discard cover packets, emit SURB acknowledgments, and buffer the recovered model fragments by index range without sender labels.

\subsection{Anonymous Acknowledgments and Retransmission}
\label{subsec:surb}
Acknowledgments support delivery confirmation, resend control, peer liveness, and fragment-completeness monitoring. Direct acknowledgments would create a privacy channel because the receiver would contact the original sender or reuse the forwarding path, exposing communication links or repeated fragment patterns.

\textit{UnlinkableDFL} uses Single-Use Reply Blocks (SURBs), a reply mechanism from mixnet and Sphinx-style anonymous communication. When creating a fragment, the sender attaches an opaque one-time return path that is independent of the forward path. The receiver can use the SURB to acknowledge delivery without learning the sender's address.

An acknowledgment is wrapped with the SURB and forwarded through the embedded return route. Forward and return paths have independently sampled hop counts and relays, so acknowledgments do not reveal stable path correlations. If an expected acknowledgment times out, the fragment is retransmitted along a fresh route with a fresh SURB, supporting delivery without direct sender identification.

\subsection{Fragmented Aggregation}
\label{subsec:fragment_fedavg}
Each node maintains a fragment pool indexed by ranges in the flattened parameter vector. The pool contains local and received fragments. Duplicate entries are removed using fragment identifiers or content hashes, and aggregation proceeds over the fragments available for each range. Similar to Shatter, aggregation is performed at the fragment level rather than over sender-attributed full models~\cite{biswas2025noiseless}. In \textit{UnlinkableDFL}, fragments arrive through the mixnet and are stored without sender identities, so aggregation consumes fragment values and ranges rather than participant labels. The default rule is \texttt{FragFedAvg}, with \texttt{FragKrum} available when Byzantine behavior is in scope~\cite{blanchard2017machine}.

\begin{algorithm}[t]
\small 
\caption{\FragFedAvgTag{} and \FragKrumTag{}}
\label{alg:fragagg}
\DontPrintSemicolon
\SetKwInOut{Input}{Input}\SetKwInOut{Output}{Output}
\Input{Local model \(W\), fragment pool \(\mathcal{P}\), rule \(\rho \in \{\mathsf{FragFedAvg},\mathsf{FragKrum}\}\), and per-range Byzantine bound \(B_r\) for \(\mathsf{FragKrum}\).}
\Output{Updated model \(\tilde W\).}
\BlankLine
Flatten \(W\) into \(x \in \mathbb{R}^n\) and initialize \(\tilde x \leftarrow x\).\;
Deduplicate \(\mathcal{P}\) by fragment identifier or content hash.\;
Let \(\mathcal{R}\) be the index ranges represented in \(\mathcal{P}\).\;
\ForEach{range \(r=[s,e) \in \mathcal{R}\)}{
    \(Y_r \leftarrow \{\mathrm{values}(f): f \in \mathcal{P},\ \mathrm{range}(f)=r\}\).\;
    \uIf{\(\rho=\mathsf{FragFedAvg}\)}{
        \(\tilde x[s:e] \leftarrow |Y_r|^{-1}\sum_{y\in Y_r} y\).\tcp*{\FragFedAvgTag}
    }
    \uElseIf{\(\rho=\mathsf{FragKrum}\)}{
        \(q_r \leftarrow |Y_r|-B_r-2\).\;
        \ForEach{candidate slice \(y \in Y_r\)}{
            Let \(N_y\) be the \(q_r\) nearest slices to \(y\) in \(Y_r\setminus\{y\}\).\;
            \(s_y \leftarrow \sum_{z\in N_y}\|y-z\|_2^2\).\;
        }
        \(y^\star \leftarrow \arg\min_{y\in Y_r} s_y\).\;
        \(\tilde x[s:e] \leftarrow y^\star\).\tcp*{\FragKrumTag}
    }
    \Else{
        \(\tilde x[s:e] \leftarrow x[s:e]\).\tcp*{fallback}
    }
}
Reshape \(\tilde x\) into model format to obtain \(\tilde W\).\;
\Return \(\tilde W\).
\end{algorithm}

Algorithm~\ref{alg:fragagg} gives the shared skeleton for \texttt{FragFedAvg} and \texttt{FragKrum}. \texttt{FragFedAvg} averages all available values for the same range. \texttt{FragKrum} scores each candidate slice by squared distance to its nearest slices and selects the lowest-score slice. Ranges absent from the pool keep their pre-aggregation local values, allowing learning to continue when routes fail, relays leave, or only partial peer fragments arrive before aggregation. Section~\ref{sec:properties} analyzes when \texttt{FragFedAvg} preserves FedAvg behavior.

\subsection{Topology Dynamics and Churn Handling}
\label{subsec:dynamic_topology}
Nodes maintain local peer views instead of a global registry. Join announcements propagate through normal control messages. Failures are detected through connectivity checks and missing acknowledgments, and inactive peers are removed from route sampling. If churn drops in-flight packets, fresh-route retransmission, fragment-pool aggregation, and local fallback for uncovered indices let learning continue without global coordination.
% Prototype Implementation
\section{Prototype Implementation}
\label{sec:implementation}
The prototype instantiates the framework as a single-machine experimental platform for repeatable DFL deployments. Each node runs learning, routing, mixing, aggregation, and local metric collection as an independent process, deployable either as a container or as a pinned OS process.\footnote{Available at: \url{https://github.com/Cyber-Tracer/DFL_PeerBasedMixing}}

\paragraph{\textbf{Communication Stack}}
Nodes communicate over a transport topology of long-lived QUIC connections implemented with \texttt{aioquic}. Each neighbor pair maintains one authenticated connection, and mutual Transport Layer Security (TLS) binds peer identity to per-run certificates rather than container addresses.

The anonymous overlay is implemented above this transport layer with the \texttt{sphinxmix} library\footnote{\url{https://pypi.org/project/sphinxmix/}}, a Python implementation of the Sphinx mix format~\cite{danezis2009sphinx}, which provides path construction, packet encryption and decryption, and SURBs. Random paths and cover packets use Python's \texttt{secrets} module. Sphinx uses the Advanced Encryption Standard (AES) in counter mode to derive per-hop keystreams, while Hash-based Message Authentication Code with SHA-256 (HMAC-SHA256) protects packet headers and payloads. All packets are padded to the same length before transmission.

\paragraph{\textbf{Learning and Aggregation Stack}}
Local training and fragmented aggregation are implemented in \texttt{PyTorch}. Datasets are loaded through \texttt{torchvision}, and model parameters are serialized, flattened, split into fragments, encoded, and reconstructed during aggregation. The evaluated default applies fragment-level 8-bit stochastic quantization followed by generic lossless compression, while a 32-bit mode keeps the payload unquantized for comparison. Local and incoming fragments are stored in a fragment pool, deduplicated, and aggregated through the \texttt{FragFedAvg} or \texttt{FragKrum} interface in Algorithm~\ref{alg:fragagg}.

\paragraph{\textbf{Scenario and Monitoring Stack}}
Scenario deployment and node orchestration use \texttt{FastAPI} for the experiment manager, which launches each node either as a Docker container or as a core-pinned OS process. The manager is outside the decentralized training path: it starts the nodes, injects per-run configuration, exposes monitoring endpoints, and exports logs, but it does not aggregate models, choose mixnet routes, or coordinate fragment exchange during training. Live metrics are streamed to the React dashboard through Server-Sent Events (SSE) and exported to comma-separated value (CSV) logs at the end of a run. The logs cover communication, learning, and resource metrics. Supplementary Material~B provides implementation configuration details and dashboard screenshots.

% Analysis
\section{Theoretical Analysis}
\label{sec:properties}
This section gives theoretical support for \textit{UnlinkableDFL}. The network-layer analysis bounds sender--message linking through route and relay uncertainty, while the learning analysis identifies when fragment-based aggregation preserves FedAvg-style behavior. The arguments draw on mixnet anonymity metrics~\cite{piotrowska2017} and fragment-level learning analysis~\cite{biswas2025noiseless}.

\subsection{Analysis Scope and Metrics}
\label{subsec:analysis_scope}
Let \(m\) denote a target fragment packet generated by sender \(S(m)\). Following the threat model in Section~\ref{sec:threat_model}, the adversary observes the network trace \(T\) defined there. In the sender-identification objective defined in Section~\ref{sec:problem_statement} and Eq.~\ref{eq:plink-objective}, the adversary wins if it identifies \(S(m)\) from the trace. The scope is deliberately network-layer: the trace does not include private local randomness, cryptographic keys, or decrypted fragment content at the destination.

A trace \(T\) usually admits multiple explanations for the same packet. This work calls each explanation a \emph{sender-labeled trajectory}. A trajectory \(\pi\) specifies a candidate sender, a feasible overlay route, and the input-output matching choices through relay outboxes. An input-output matching is the association between a packet entering a relay and one packet later emitted from that relay's shuffled outbox. It is hidden when this association is not visible in \(T\). Let \(\Pi_T(m)\) be the set of such trajectories that remain feasible after conditioning on \(T\). For a candidate sender \(u\), let
\begin{equation}
    \Gamma_T(u,m)=\{\pi\in\Pi_T(m): \pi \text{ starts at } u\}
\end{equation}
be the trajectories that attribute \(m\) to \(u\). The trace-consistent candidate sender set is \(\mathcal{U}_T=\{u:\Gamma_T(u,m)\neq\emptyset\}\). Given this set, define
\begin{equation}
\label{eq:plink-def}
    p_{\mathrm{link}}(m\mid T)
    =\max_{u\in\mathcal{U}_T}\Pr[S(m)=u\mid T]
\end{equation}
as the adversary's best conditional sender-identification probability, the quantity in the objective of Eq.~\ref{eq:plink-objective}. The advantage in Eq.~\ref{eq:adversarial_advantage} is non-negative because, for each trace, the maximum posterior over \(\mathcal{U}_T\) is at least the uniform value \(1/|\mathcal{U}_T|\). Finally,
\begin{equation}
    \gamma_T=\max_{u\in\mathcal{U}_T}|\Gamma_T(u,m)|
\end{equation}
records the largest number of feasible trajectories assigned to any one candidate sender. Moreover, path uncertainty and sender uncertainty are not identical: one sender can be consistent with several routes and relay matchings.

To keep the guarantee and evaluation metrics aligned, we use three entropy quantities. First, let \(\mathcal{H}(\pi)\) be the set of hidden relay shuffles crossed by trajectory \(\pi\), and let \(h(\pi)=|\mathcal{H}(\pi)|\). Using the configured outbox size \(O\), the \emph{relay entropy} is
\begin{equation}
\label{eq:relay-entropy}
H_{\mathrm{relay}}(\pi)
= h(\pi)\log_2 O
=-\log_2 O^{-h(\pi)}.
\end{equation}
This is the min-entropy contribution of local input-output matching under uniform outbox shuffling. Section~\ref{sec:eval-entropy} also reports \(H_{\mathrm{relay}}^{\mathrm{loc}}\), the Shannon entropy of one relay's absorption-or-forwarding distribution. We keep the superscript to distinguish that diagnostic from the min-entropy contribution in Eq.~\ref{eq:relay-entropy}.

Second, let \(R_K\) denote the random route sampled under maximum path length \(K\). After the adversary observes trace \(T\), \(\mathcal{R}_K(T)\) denotes the set of route values of \(R_K\) that remain feasible, and \(h(r)\) is the number of hidden shuffles on route \(r\). The count-based \emph{path entropy} is
\begin{equation}
\label{eq:path-count-entropy}
H_{\mathrm{path}}^{\mathrm{cnt}}(m\mid T)
=\log_2\sum_{r\in\mathcal{R}_K(T)} O^{h(r)} .
\end{equation}
The evaluation reports \(H_{\mathrm{path}}\) as a shorthand special case of \(H_{\mathrm{path}}^{\mathrm{cnt}}\). Under the uniform route-counting setup used in Section~\ref{sec:eval-entropy}, where hidden paths of length \(1\) to \(K\) are counted with common outbox size \(O\), this special case is \(H_{\mathrm{path}}=\log_2\sum_{k=1}^{K}O^k\). This quantity counts feasible shuffled paths under uniform choices.

Third, the privacy guarantee uses posterior trajectory min-entropy. Here, posterior means after conditioning on the observed trace \(T\):
\begin{equation}
\label{eq:trajectory-min-entropy}
H_{\infty}(\Pi_m \mid T)
= -\log_2 \max_{\pi \in \Pi_T(m)}
\Pr[\Pi_m=\pi \mid T].
\end{equation}
The relationship is exact under a uniform posterior over the counted trajectories, in which case \(H_{\infty}(\Pi_m\mid T)=H_{\mathrm{path}}^{\mathrm{cnt}}(m\mid T)\). In general, \(H_{\infty}(\Pi_m\mid T)\le H_{\mathrm{path}}^{\mathrm{cnt}}(m\mid T)\), so the guarantee below is stated in terms of min-entropy, while relay/path entropy serve as uniform-case parameter metrics for interpreting \(K\) and \(O\).

\subsection{Network-Layer Unlinkability Guarantee}
\label{subsec:unlinkability}
In \textit{UnlinkableDFL}, externally visible packets use a uniform Sphinx-style format and padding, and cover packets follow the same outbox handling as real fragment packets. Forward paths are sampled independently across fragments, while SURB acknowledgments use separately encoded return paths. These choices remove packet appearance, type, and acknowledgment direction as direct sender-dependent features in \(T\).

\begin{theorem}[Network-layer sender--message unlinkability]
\label{thm:network_unlinkability}
Consider a fragment packet \(m\) whose sender-labeled trajectories \(\Pi_T(m)\) are consistent with the adversary trace \(T\). Let \(h=\min_{\pi\in\Pi_T(m)}h(\pi)\) be the minimum number of hidden forwarding shuffles on any feasible trajectory. Assume that visible packet features are sender-independent and that each hidden forwarding shuffle uses a uniformly shuffled outbox of size \(O\). If the adversary uses only network-layer observations, its best sender-linking probability is bounded by
\begin{equation}
\label{eq:sender-link-bound}
p_{\mathrm{link}}(m \mid T)
 \le \min\!\left\{1,\gamma_T\,2^{-H_{\infty}(\Pi_m \mid T)}\right\}.
\end{equation}
Moreover, if the hidden shuffle choices remain unrevealed except through their feasible outbox positions,
\begin{equation}
\label{eq:path-min-entropy}
H_{\infty}(\Pi_m \mid T)
\ge H_{\infty}(R_K \mid T)
+ h\log_2 O,
\end{equation}
where \(R_K\) is the random route variable defined above, and \(H_{\infty}(R_K\mid T)\) is the route uncertainty that remains after observing \(T\). Combining Eqs.~\ref{eq:sender-link-bound} and~\ref{eq:path-min-entropy} gives the explicit parameter form
\begin{equation}
\label{eq:explicit-ko-bound}
p_{\mathrm{link}}(m \mid T)
 \le \min\!\left\{1,
 \gamma_T 2^{-H_{\infty}(R_K\mid T)} O^{-h}\right\}.
\end{equation}
Here \(h\le K\). The relay-shuffle contribution is therefore at most \(K\log_2 O\), reached only when the least-hidden feasible trajectory still hides all \(K\) forwarding stages.
\end{theorem}

Operationally, the theorem turns min-entropy into a best-guess bound. If the trace leaves \(b\) bits of trajectory min-entropy, the largest posterior mass of a single trajectory is at most \(2^{-b}\), and sender linking pays only the multiplicity factor \(\gamma_T\). The parameter form exposes the controls used later: \(K\) enlarges route uncertainty, \(O\) adds \(\log_2 O\) bits for each hidden shuffle, and Corollary~\ref{cor:compromised} accounts for partial relay compromise.

The bound is conditional on two idealizations: sender-independent visible features and uniformly shuffled outboxes. The fill-and-delay outbox in Section~\ref{subsec:fragment_outbox} approximates these conditions, so the result should be read as a parameter guide, not as an end-to-end anonymity proof or a prediction of concrete attack success. Section~\ref{sec:eval-unlink} measures that success relative to the random baseline \(1/|\mathcal{U}_T|\), and Supplementary Material~A.1 gives the proof.

\begin{remark}[Timing as a residual channel]
The bound counts route and relay-shuffle uncertainty under the assumption that visible packet features, including timing, are sender-independent. The configured delay jitter is small and primarily decorrelates local training from outgoing transmission rather than obfuscating per-hop timing, so this assumption is an idealization of the timing channel. Residual timing and volume leakage are measured directly by the attack in Section~\ref{sec:eval-unlink} rather than assumed away.
\end{remark}

\begin{corollary}[Compromised-relay regime]
\label{cor:compromised}
Suppose each forwarding stage samples a relay independently and uniformly from a population in which a fraction \(c\) is compromised. A length-\(K\) trajectory then has all forwarding stages compromised with probability \(c^{K}\), so with probability \(1-c^{K}\) it retains at least one hidden shuffle and \(p_{\mathrm{link}}(m\mid T)\le \gamma_T 2^{-H_{\infty}(R_K\mid T)}O^{-1}\). The expected number of hidden shuffles is \(K(1-c)\). This idealized calculation is replaced in topology-constrained deployments by the corresponding probability that all forwarding stages are compromised.
\end{corollary}

Supplementary Material~A.2 handles acknowledgments by treating SURB return paths as independent of the forward trajectory. The guarantee remains conditional on the network-layer trace, so endpoint compromise or content-based linking is outside this theorem. It also assumes a passive adversary. Active attacks such as flooding, tagging, and dropping lie outside the bound. Sphinx integrity checks, cover packets, and randomized multi-hop release can make such attacks less direct, but a formal active-adversary analysis is future work.

\subsection{Learning Preservation of \texttt{FragFedAvg}}
\label{sec:convergence}
Fragmented aggregation preserves FedAvg's index-wise averaging behavior when parameter ranges receive sufficient peer coverage, with additional terms capturing incomplete coverage, fallback, and quantization. We use the standard assumptions for FedAvg-like methods: each local objective \(F_k\) is \(L\)-smooth, stochastic gradients are unbiased with variance bounded by \(\sigma_g^2\), and local model drift is controlled~\cite{mcmahan2017communication,karimireddy2020scaffold}. Let
\begin{equation}
    F(w)=\frac{1}{N}\sum_{k=1}^{N} F_k(w).
\end{equation}
In round \(t\), peer \(k\) performs \(\tau\) local SGD steps with step size \(\eta\), producing a shareable model state \(w^{(k)}_t\). A standard FedAvg-style update averages each parameter over the participating set \(\mathcal{C}_t\):
\begin{equation}
    \bar w_{t,i}
    = \frac{1}{|\mathcal{C}_t|}
      \sum_{k\in\mathcal{C}_t} w^{(k)}_{t,i}.
\end{equation}

\texttt{FragFedAvg} performs the same averaging by index range. Let \(P_t\) be the fragment pool at a node, let \(c_{t,i}\) be the number of available fragments covering parameter index \(i\), and let \(Q(\cdot)\) denote the stochastic quantizer used before lossless compression. We assume \(Q\) is unbiased and has bounded quantization variance \(\nu_q\), while compression is lossless. The aggregated value is
\begin{equation}
    \tilde w_{t,i}
    = \frac{1}{c_{t,i}}
      \sum_{f\in P_t:\, i\in f} Q(f[i]),
\end{equation}
with fallback to the pre-aggregation local value when \(c_{t,i}=0\).

\begin{theorem}[Learning preservation of \texttt{FragFedAvg}]
\label{thm:fragfedavg_convergence}
Let \(p_{t,i}=\Pr(c_{t,i}\ge 1)\), let \(\bar p=\min_i\inf_t p_{t,i}\), and let \(\bar K_{\mathrm{eff}}\) be the minimum expected number of contributors among indices with coverage. If the fallback deviation is bounded by \(\Delta_{\mathrm{fb}}\) and \(\eta \le 1/(2L)\), then the fragment-based iterate satisfies
\begin{equation}
\label{eq:frag-rate}
\begin{split}
    \frac{1}{T}\sum_{t=0}^{T-1}
    \mathbb{E}\|\nabla F(\tilde w_t)\|^2
    \le
    O\Big(
        \frac{1}{\eta T}
        + \eta L(\sigma_g^2+\nu_q)
        + \frac{1}{\bar p\,\bar K_{\mathrm{eff}}\tau}
        {} \\
        \qquad\qquad
        + (1-\bar p)\Delta_{\mathrm{fb}}
    \Big).
\end{split}
\end{equation}
\end{theorem}

Supplementary Material~A.3 (``Learning Preservation of \texttt{FragFedAvg}'') proves the bound. The first two terms match the usual optimization and stochastic-gradient terms, with \(\nu_q\) capturing the additional variance introduced by unbiased quantization. The third term reflects reduced effective participation caused by incomplete fragment coverage. The final term captures fallback bias when an index receives no peer fragment before aggregation. When \(\bar p=1\) and \(\nu_q=0\), \texttt{FragFedAvg} reduces to FedAvg over the same participating models. As coverage increases and quantization variance remains small, the extra terms become secondary. The analysis assumes the pool aggregates fragments produced for the current round. Cross-round fragment staleness under fully asynchronous schedules is not modeled and is left to future work.

\subsection{Robustness of \texttt{FragKrum}}
\label{subsec:fragkrum_analysis}
Anonymous aggregation removes sender labels from the fragment pool, which also removes simple sender-based reputation or filtering. Here, Byzantine slices are malicious or faulty fragment values that may be arbitrary rather than honest training outputs. \texttt{FragKrum} addresses this setting by applying Krum independently to the values available for each index range. Let \(Y_r\) be the candidate fragment slices for range \(r\), and let \(B_r\) be an upper bound on Byzantine slices in \(Y_r\). Here \(B_r\) is a configured robustness budget set from the assumed adversary fraction, analogous to the parameter \(f\) in standard Krum, rather than a quantity estimated online from the anonymous pool.

\begin{theorem}[Fragment-level Krum robustness]
\label{thm:fragkrum_robustness}
For an index range \(r\), assume \(Y_r\) contains \(n_r\) candidate slices, at most \(B_r\) of them are Byzantine, and \(n_r > 2B_r+2\). If honest slices for \(r\) have bounded diameter and Byzantine slices are arbitrary, the \texttt{FragKrum} rule in Algorithm~\ref{alg:fragagg} selects a slice whose Krum score is no larger than that of any honest candidate. Consequently, the selected slice lies in the robust neighborhood characterized by the standard Krum condition~\cite{blanchard2017machine}, without requiring sender identities.
\end{theorem}

For each range, \texttt{FragKrum} constructs the same nearest-neighbor score used by Krum, but the candidates are fragment slices rather than full model vectors. The condition \(n_r>2B_r+2\) ensures that the nearest-neighbor set of an honest slice contains enough honest candidates to dominate the score. Since the score depends only on slice distances, sender labels are not used. Supplementary Material~A.4 (``Robustness of \texttt{FragKrum}'') gives the reduction from range-wise slices to the standard Krum argument and discusses partial coverage.

The robustness guarantee is therefore local to covered ranges. If a range has too few candidates or too many Byzantine slices, \texttt{FragKrum} falls back to the same limitation as Krum under insufficient honest majority.
% Experiments
\section{Evaluation}
\label{sec:evaluation}
This section evaluates the prototype in Section~\ref{sec:implementation}. The evaluation first measures learning utility as the deployment scales, then quantifies how mixing, cover packets, randomized delays, and path length shape network-layer sender--message linking. The remaining experiments report mixnet cost, payload handling, churn and Byzantine robustness, and a curious-recipient boundary test.

\subsection{Experimental Setup}
\label{sec:exp-setup}

Each node runs training, route sampling, mixing, acknowledgment handling, fragment storage, aggregation, and local metric collection. The prototype supports two interchangeable single-machine experimental deployments selected by configuration: one container per node over a virtual network, or one OS process per node pinned to a core over loopback. Both use authenticated QUIC links between transport neighbors, so the deployment mode does not change the learning, routing, or unlinkability mechanisms. Unless varied, experiments use \(N=32\), a degree-4 circulant topology, \(K=2\), \(O=100\), \(\mu=0.1\)~s, \(\sigma=0.001\), LeNet-5, 10 rounds, Dirichlet partitioning with \(\alpha=10\), 8-bit stochastic quantization with lossless compression, and \texttt{FragFedAvg}.

\begin{table}[!h]
\centering
\caption{Summary of evaluation settings.}
\label{tab:eval-settings}
\begin{tabular}{@{}>{\raggedright\arraybackslash}p{0.25\linewidth}@{\hspace{0.35em}}>{\raggedright\arraybackslash}p{0.22\linewidth}@{\hspace{0.35em}}>{\raggedright\arraybackslash}p{0.49\linewidth}@{}}
\hline
\textbf{Experiment} & \textbf{Variable} & \textbf{Value} \\
\hline
Learning Utility & \(N\) & 16, 32, 64, 100 \\
 & Dataset & MNIST, Fashion-MNIST, CIFAR-10 \\
 & Model & LeNet-5 \\
 & Topology & degree-4 circulant \\
\hline
Linking attack & Ablation & direct, no-cover, no-delay, \(K{=}1\) \\
\hline
Mixnet parameters & Path length \(K\) & 1, 2, 3, 4, 5 \\
 & Outbox size \(O\) & 10, 25, 50, 100, 200 \\
 & Delay mean \(\mu\) & 0.04, 0.10, 0.32~s \\
\hline
Payload encoding & Encoding & 8-bit, 32-bit \\
 & Compression & lossless \\
 & Model & LeNet-5, SqueezeNet, MobileNetV2 \\
\hline
Partial exchange & Ratio \(r\) & 0.10, 0.25, 0.50, 0.75, 1.00 \\
\hline
Content attack & Heterogeneity \(\alpha\) & 0.1, 0.5, 1, 10 \\
 & Ratio \(r\) & 0.10, 0.25, 0.50, 0.75, 1.00 \\
 & Target & identity linking, bucket tracking \\
\hline
Stressors & Churn & node join, node exit \\
 & Byzantine attack & label-flip, Gaussian noise \\
 & Aggregator & \texttt{FragFedAvg}, \texttt{FragKrum} \\
\hline
\end{tabular}
\end{table}

Experiments ran on an Ubuntu server with an AMD EPYC 7502P processor, 50 visible cores, and 94~GiB memory. Each node is allocated one vCPU-equivalent for repeatability, a container CPU quota under the container deployment or a pinned core under the process deployment, limiting peak CPU use per node without reserving a physical core. Table~\ref{tab:eval-settings} lists the variables and values.

\subsection{Learning Utility at Scale}
\label{sec:eval-scale}

\begin{figure}[t]
    \centering
    \includegraphics[width=\linewidth]{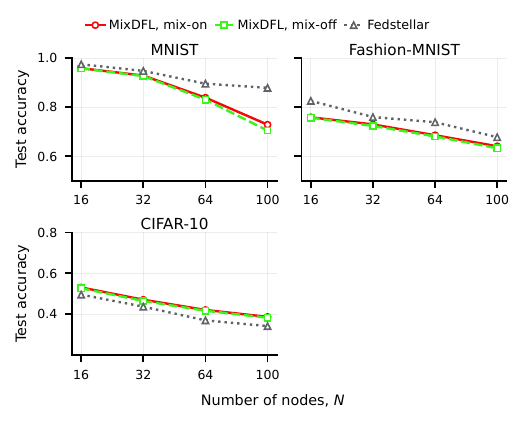}
    \caption{Test accuracy as the network scales, for MNIST, Fashion-MNIST, and CIFAR-10 (LeNet-5).}
    \label{fig:eval-scale-accuracy}
\end{figure}

This experiment measures learning utility as \(N\) grows and each peer holds less data. It uses the default topology and mixnet parameters, LeNet-5, ten rounds, Dirichlet partitioning with \(\alpha=10\), 8-bit stochastic quantization with lossless compression, and \texttt{FragFedAvg}. Here \(\alpha\) controls data heterogeneity, with larger values closer to independent and identically distributed (IID) partitions. Figure~\ref{fig:eval-scale-accuracy} reports round-10 test accuracy against the non-anonymous Fedstellar baseline~\cite{MARTINEZBELTRAN2024122861}, and Supplementary Material~C.1 gives the full convergence curves.

\textit{UnlinkableDFL} scales to \(N=100\) on all three tasks while keeping fragment completeness at 100\%. Accuracy decreases with \(N\), from 0.959 to 0.729 on MNIST, 0.759 to 0.641 on Fashion-MNIST, and 0.531 to 0.387 on CIFAR-10, because each node trains on a smaller data partition. The mixnet-on and mixnet-off curves closely match, so the mixnet does not measurably degrade accuracy. Relative to the non-anonymous Fedstellar baseline the gap depends on the dataset: \textit{UnlinkableDFL} is ahead on CIFAR-10 (0.387 versus 0.342 at \(N=100\)) but behind on MNIST and Fashion-MNIST (0.729 versus 0.879 on MNIST at \(N=100\)). The framework thus scales to data-sparse deployments without the mixnet degrading accuracy, while its standing against a non-anonymized platform varies by task.

\subsection{Network-Layer Unlinkability Under Attack}
\label{sec:eval-unlink}

This experiment estimates the unlinkability advantage of Eq.~\ref{eq:adversarial_advantage} in the open world, where the candidate set \(\mathcal{U}_T\) contains all trace-consistent senders. A passive adversary tries to link a target model-fragment packet to its sender, and the privacy claim is checked by attack accuracy rather than entropy alone. It pairs this attack with the entropy parameters \(K\) and \(O\), complementing the analysis in Section~\ref{sec:properties}.

\paragraph{\textbf{Entropy parameters}}
\label{sec:eval-entropy}
We first summarize how the two controls change uncertainty before relating them to attack success. The reported \(H_{\mathrm{path}}\) is the uniform-counting special case of \(H_{\mathrm{path}}^{\mathrm{cnt}}\) from Section~\ref{subsec:analysis_scope}. It counts feasible shuffled paths under uniform route sampling and relay shuffling, rather than estimating an adversarial posterior. For the relay-level view, we use \(H_{\mathrm{relay}}^{\mathrm{loc}}=-p_{\mathrm{abs}}\log_2 p_{\mathrm{abs}}-O p_{\mathrm{fwd}}\log_2 p_{\mathrm{fwd}}\), where \(p_{\mathrm{abs}}=1/K\) and \(p_{\mathrm{fwd}}=(K-1)/(KO)\). This metric describes uncertainty inside one relay outbox, while Theorem~\ref{thm:network_unlinkability} bounds posterior sender-linking after conditioning on the network trace.

\begin{figure}[!t]
    \centering
    \includegraphics[width=\linewidth]{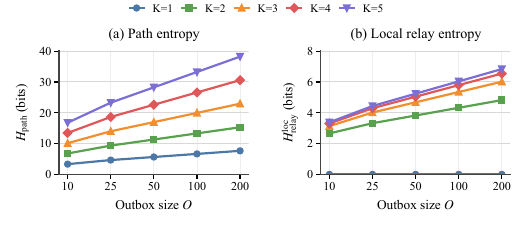}
    \caption{Path entropy (left) and per-relay outbox entropy (right) as \(K\) and \(O\) vary.}
    \label{fig:eval-entropy-structure}
\end{figure}

The direct entropy controls are the maximum path length \(K\) and the outbox size \(O\). The counted path space scales as \(O+O^2+\cdots+O^K\). At \(O=100\), increasing \(K\) from 1 to 5 raises \(H_{\mathrm{path}}\) from 6.64 to 33.23 bits, while \(H_{\mathrm{relay}}^{\mathrm{loc}}\) rises from 0 to 6.04 bits because \(K=1\) has no relay-level forwarding ambiguity. At \(K=2\), increasing \(O\) from 10 to 100 raises \(H_{\mathrm{path}}\) from 6.78 to 13.30 bits and \(H_{\mathrm{relay}}^{\mathrm{loc}}\) from 2.66 to 4.32 bits. Thus \(K\) is the coarse entropy lever, while \(O\) tunes the relay-shuffle contribution inside a chosen path length.

Cover packets support these levers. They keep the outbox population near \(O\) even when a node holds few real fragments, so the per-hop \(1/O\) shuffle uncertainty is maintained rather than collapsing to the real-queue length. They also reduce packet volume as a sender signal. As the ablation below shows, cover packets add a small privacy margin on top of multi-hop mixing rather than driving the protection.

Topology and degree bound the trace-consistent candidate set over which uncertainty can be distributed. At the default \(N=32\), \(K=2\), and degree-4 circulant topology, a node reaches 12 peers within two hops, while a full mesh exposes \(N-1=31\). Raising degree expands this set until saturation. A diameter-oriented circulant saturates with fewer links than a local ring lattice, but may concentrate traffic on fewer transport edges. Supplementary Material~C.2 provides degree, topology, relay-entropy, and cost details.

\paragraph{\textbf{Sender-linking attack and ablation}}
No off-the-shelf attack fits peer-run DFL, so we build a passive Bayesian adversary that estimates the bounded quantity \(p_{\mathrm{link}}(m\mid T)=\max_{u\in\mathcal{U}_T}\Pr[S(m)=u\mid T]\) by fusing three classical mixnet signals: a \emph{predecessor} channel from the visible last hop and any compromised-relay segment, a \emph{timing} channel that maps the arrival time back through the delay law to an emission window, and a \emph{volume} channel from above-cover outbound activity. Suspects are the trace-consistent origins within the hop budget, the maximum-a-posteriori one is the guess, and we report top-1 accuracy against the random baseline \(1/|\mathcal{U}_T|\). This baseline differs across ablations because the trace-consistent candidate set \(\mathcal{U}_T\) depends on the hop budget, shrinking to the degree-four neighborhood at \(K{=}1\). The adversary sees only the recorded edge-level trace, in which Sphinx re-randomization blocks cross-hop linking, while the true origin stays sealed in the onion for scoring only. We ablate the mechanisms (direct, no-cover, no-delay, \(K{=}1\), full) and sweep \(K\), \(O\), and the compromised fraction \(c\) at \(N=32\). Supplementary Material~C.2 gives the trace fields, feature definitions, and the attack procedure.

\begin{table}[t]
\centering
\caption{Sender-linking ablation at \(N=32\) on MNIST: attack accuracy (mean \(\pm\) standard deviation), random baseline, and CPU time per round.}
\label{tab:eval-ablation}
\begin{tabular}{@{}lccc@{}}
\hline
\textbf{Configuration} & \textbf{Attack acc.} & \textbf{Baseline} & \textbf{CPU (s/round)} \\
\hline
Direct (no mixnet) & \(0.99\pm0.01\) & 0.08 & 6.6 \\
Single hop (\(K{=}1\)) & \(1.00\pm0.00\) & 0.25 & 9.5 \\
No cover packets & \(0.10\pm0.02\) & 0.08 & 9.7 \\
No randomized delay & \(0.06\pm0.02\) & 0.08 & 296.7 \\
Full \textit{UnlinkableDFL} & \(0.08\pm0.01\) & 0.08 & 14.2 \\
\hline
\end{tabular}
\end{table}

\begin{figure}[t]
    \centering
    \includegraphics[width=\linewidth]{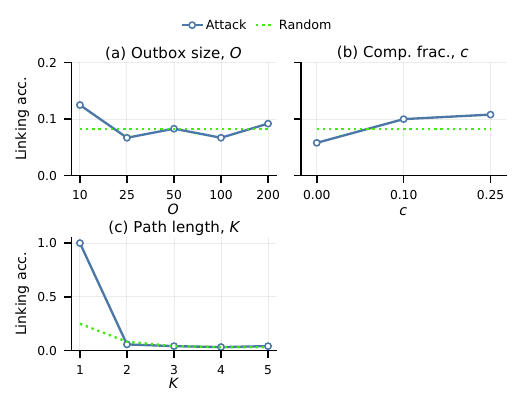}
    \caption{Sender-linking accuracy across the \(K\), \(O\), and \(c\) sweeps at \(N=32\), against the random baseline.}
    \label{fig:eval-linking-attack}
\end{figure}

Table~\ref{tab:eval-ablation} reports the ablation results, and Figure~\ref{fig:eval-linking-attack} presents the parameter sweeps. The three mechanisms play distinct roles. \emph{Multi-hop mixing} is the primary privacy lever. Without the mixnet, direct DFL is almost fully linkable (0.99) because the last transport hop is the sender, and a single hop (\(K=1\)) is equally exposed (1.00) since the first relay is the sender's neighbor. With two or more hops the attack collapses to the random baseline, where the full configuration matches the 0.08 baseline (\(0.08\pm0.01\)), accuracy drops sharply at \(K=2\) and stays there through \(K=5\), and the outbox size \(O\) has no further effect once \(K\ge 2\). \emph{Cover packets} add a small privacy margin. Removing them raises linking to 0.10, just above the 0.08 baseline, because a node that transmits only when it holds real fragments leaves a weak volume signal that cover-packet padding otherwise masks. \emph{Randomized delay} does not change linkability under this passive attack, since removing it leaves the attack at 0.06, near the baseline, but it is what keeps the cost bounded. Here removing the delay means releasing at a near-zero interval rather than under a rate-limited scheduler, so the outbox keeps filling empty slots with cover packets and floods the network. CPU therefore rises from 14.2 to 296.7~s/round while the attack does not improve. The one lever that helps the adversary is relay compromise, where accuracy rises from 0.06 to 0.11 as the compromised fraction grows from 0 to 0.25, since fewer hidden shuffles remain when more relays on a path are compromised. The CPU column quantifies the cost of this protection, with the full configuration reaching near-baseline unlinkability at 14.2~s/round against 6.6~s/round for direct DFL.

\subsection{Unlinkability--Cost Trade-off}
\label{sec:eval-cost}

Enabling the mixnet sharply increases the per-round cost while leaving the learning result close to the mixnet-off run. We report two distinct quantities: CPU time, the compute consumed per round, and wall-clock round duration, the elapsed time per round, which differ because randomized release and buffering add waiting that does not consume compute. At \(N=32\) on MNIST under the default degree-four topology, with \(K=2\), \(O=100\), and \(\mu=0.1\), turning mixing on raises CPU time from 6.6 to 14.2~s/round and wall-clock round duration from 6.3 to 40.5~s, and increases the bytes a node sends per round several-fold as cover-packet padding replaces direct delivery, from about 12 to 60~MB on MNIST. Wall-clock round duration grows faster than CPU time because randomized release and buffering inflate elapsed time without proportional compute. This is the dominant prototype cost of unlinkability, and its dependence on graph degree and network size is reported in Supplementary Material~C.2.

\begin{table}[t]
\centering
\caption{Cost for the mixnet parameters at \(N=32\).}
\label{tab:eval-parameters}
\setlength{\tabcolsep}{2.5pt}
\begin{tabular}{@{}cccccc@{}}
\hline
\textbf{Parameter} & \textbf{Value} & \(H_{\mathrm{path}}\) (bits) & \textbf{RTT (s)} & \textbf{CPU (s/round)} & \textbf{Mem. (MB)} \\
\hline
\(\mu\) & 0.04 & 13.30 & 12.0 & 20.0 & 393 \\
        & 0.10 & 13.30 & 29.7 & 14.2 & 393 \\
        & 0.32 & 13.30 & 94.4 & 11.8 & 393 \\
\hline
\(K\)  & 1 & 6.64 & 7.4 & 9.5 & 390 \\
        & 2 & 13.30 & 29.7 & 14.2 & 393 \\
        & 5 & 33.23 & 245.1 & 73.7 & 405 \\
\hline
\(O\)  & 10 & 6.78 & 12.7 & 10.0 & 391 \\
        & 100 & 13.30 & 29.7 & 14.2 & 393 \\
\hline
\end{tabular}
\end{table}

Once mixing is enabled, the mixnet parameters also shape the runtime cost. Table~\ref{tab:eval-parameters} summarizes this parameter trade-off at \(N=32\). Path length \(K\) gives the strongest entropy gain, but dominates RTT and CPU cost. Increasing \(O\) gives a smaller entropy gain with modest CPU and memory changes, while \(\mu\) mainly shifts delay and cover-packet pressure. These research-prototype latencies make moderate \(K\) the practical region for rounds that tolerate tens of seconds. Total network bytes including cover packets, and the per-node message load as the deployment scales, are reported in Supplementary Material~C.2.

\subsection{Payload Handling and Partial Updates}
\label{sec:eval-payload}

Payload handling is the main scalability mechanism outside the mixnet. With a 10240-byte Sphinx body, a 32-bit LeNet-5 update requires 26 fragments, while the 8-bit encoded version requires 7. This 3.7\(\times\) fragment-count reduction is close to the expected 4\(\times\) byte saving. Larger models make this more important, as SqueezeNet and MobileNetV2 otherwise push per-round traffic into hundreds of megabytes or more at high \(N\). Supplementary Material~C.1 reports the per-model fragment counts and per-round bytes for LeNet-5, SqueezeNet, and MobileNetV2. In the convergence corpus, 8-bit and 32-bit trajectories stay within 0.16 percentage points on MNIST and 0.06 on Fashion-MNIST, as reflected by Figure~\ref{fig:eval-payload}(a). Figure~\ref{fig:eval-payload}(c) shows the corresponding measured workload reduction on both datasets, from about 204 to about 54 meaningful model-fragment packets per peer and round.

\begin{figure}[t]
    \centering
    \includegraphics[width=\linewidth]{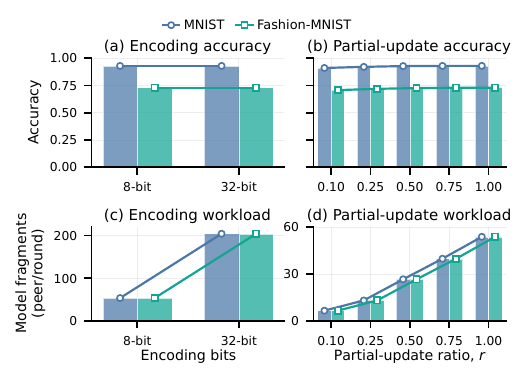}
    \caption{Payload handling at \(N=32\). Top: accuracy under encoding and partial-update choices. Bottom: meaningful fragments per peer-round by dataset.}
    \label{fig:eval-payload}
\end{figure}

Figure~\ref{fig:eval-payload}(b) reports partial exchange. On MNIST, round-10 aggregated accuracy rises from 91.00\% at \(r=0.10\) to 92.94\% at \(r=1.00\). Fashion-MNIST follows the same trend, increasing from 70.69\% to 72.96\%. Figure~\ref{fig:eval-payload}(d) shows the expected workload scaling on both datasets, from about 6.6 to about 53.8 meaningful model-fragment packets per peer and round. Because the mixnet pads outboxes with cover packets, reducing \(r\) does not directly lower the observed packet rate. It instead shortens the meaningful exchange workload and reduces delivered model state. Supplementary Material~C.1 reports fragment counts, bytes per round, quantization, and full partial-update trajectories.

\subsection{Churn and Byzantine Robustness}
\label{sec:eval-stressors}

This experiment evaluates fragment exchange under peer churn and malicious fragment values. Churn uses the default \(N=32\), degree-four circulant topology, and \(K=2\). Byzantine robustness uses \(N=20\) on a full mesh so every aggregator observes the same corrupted candidates.

\paragraph{\textbf{Churn}} We define two churn scenarios: late join and peer exit.
\begin{itemize}
    \item \textit{Late join:} one peer is absent from rounds 1--2 and starts participating in round 3. The stable-peer curve excludes the joiner and shows how quickly the new peer aligns after receiving fragments.
    \item \textit{Peer exit:} one peer leaves after round 5. Remaining peers are grouped by exposure: \textit{direct} transport neighbors lose one QUIC neighbor, \textit{overlay} peers lose a two-hop route, and \textit{unaffected} peers are outside the departed peer's two-hop exchange region.
\end{itemize}

\begin{figure}[!t]
    \centering
    \includegraphics[width=\linewidth]{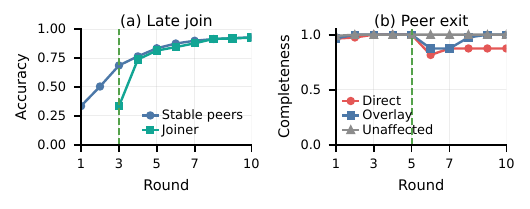}
    \caption{Churn at \(N=32\). Left: accuracy after one peer joins at round 3. Right: fragment completeness after one peer exits.}
    \label{fig:eval-churn}
\end{figure}

Figure~\ref{fig:eval-churn} reports both runs. Fragment completeness is the fraction of expected peer fragments received in a round, normalized by the stable no-churn plan. The late joiner catches up within two rounds, so joining does not require a global restart. After an exit, \textit{unaffected} peers remain stable, \textit{overlay} peers recover after rerouting, and \textit{direct} transport neighbors stay below full completeness because some strict two-hop routes no longer exist. Churn therefore degrades coverage locally and temporarily.

\paragraph{\textbf{Byzantine robustness}}
For anonymous Byzantine fragments, \(b\) malicious peers are selected and their corrupted model values are fragmented and delivered without sender labels. Following DFL robustness stress tests such as DART~\cite{feng2024dart}, we sweep \(b=0,\ldots,10\) at \(N=20\) over ten rounds with two attacks. Under \textit{label flip}, malicious peers train on flipped labels before producing local models. Under \textit{Gaussian noise}, they add zero-mean noise with standard deviation ten times the honest-update standard deviation.

Figure~\ref{fig:eval-krum} compares \texttt{FragFedAvg} and \texttt{FragKrum}. The green dashed boundary marks \(b=8\), the largest Byzantine count satisfying \(N>2b+2\) for \(N=20\). \texttt{FragFedAvg} degrades under label flip and collapses under high-magnitude Gaussian noise because one corrupted vector can dominate the mean. \texttt{FragKrum} stays near the clean run through the covered range and often beyond it, except for the expected failure at the largest label-flip count. This supports \texttt{FragKrum} as an optional robustness mode for anonymous fragment pools. Supplementary Material~C.3 gives supporting plots.

\begin{figure}[t]
    \centering
    \includegraphics[width=\linewidth]{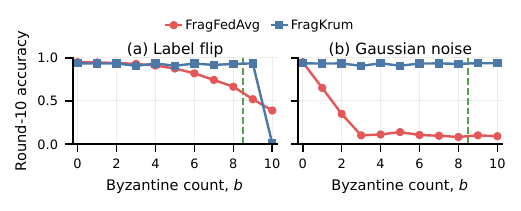}
    \caption{Accuracy under Byzantine fragments at \(N=20\), full mesh. Green dashed line: the Krum coverage limit.}
    \label{fig:eval-krum}
\end{figure}

\subsection{Residual Content-Based Linkability}
\label{sec:eval-content-linking}

The framework's privacy claim is deliberately network-layer. We therefore include a curious-recipient attack that operates after fragments are decrypted by their destination. This gives the attacker more information than the adversary in Section~\ref{sec:threat_model}, so the experiment is a boundary test rather than the guarantee evaluated in Theorem~\ref{thm:network_unlinkability}.

\begin{figure}[t]
    \centering
    \includegraphics[width=\linewidth]{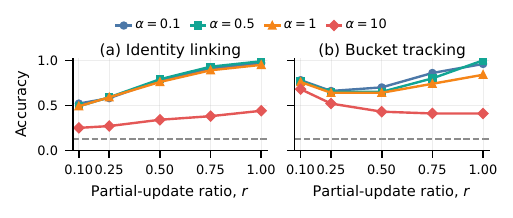}
    \caption{Curious-recipient content attack. Left: linking buckets to identities. Right: grouping same-sender fragments.}
    \label{fig:eval-content-attack}
\end{figure}

\begin{itemize}
    \item \textit{Recipient view:} the recipient gets a full first-round update as a reference. Other peers aggregate fully, while the recipient receives partial fragments by the ratio \(r\).
    \item \textit{Content comparison:} from round two, the recipient subtracts each decrypted fragment from its local model and compares fragment differences by cosine distance.
    \item \textit{Targets:} \textit{bucket tracking} groups fragments from the same hidden sender across rounds. \textit{Identity linking} further assigns each bucket to a concrete sender identity.
\end{itemize}

Figure~\ref{fig:eval-content-attack} shows that model content can remain linkable after network anonymization. Smaller \(\alpha\) and larger \(r\) strengthen payload fingerprints. Full updates reach 0.95--0.99 identity-linking accuracy for \(\alpha\le 1\), while near-IID data at \(\alpha=10\) drops to about 0.44. At \(\alpha=0.1\), reducing \(r\) from 1.00 to 0.10 lowers identity linking from 0.97 to 0.52 and bucket tracking from 0.96 to 0.78. Thus, \textit{UnlinkableDFL} protects network-layer sender--message links for model-fragment packets, while payload-layer fingerprints remain outside its defense.
% Conclusion
\section{Conclusion}
\label{sec:conclusion}
This paper presented \textit{UnlinkableDFL}, a DFL framework that addresses network-layer linkability in peer-to-peer model sharing. Each participant acts as both learner and mix relay: model states are fragmented into uniform onion-encrypted packets, carried through independently sampled mixnet paths with cover packets and randomized delays, and aggregated locally from a fragment pool without sender identities. The analysis formalizes the protected sender--message link, relates path and relay uncertainty to posterior linking advantage, and shows when fragmented aggregation preserves FedAvg-style behavior, with \texttt{FragKrum} available for anonymous Byzantine fragments.

The prototype results show that this protection comes with explicit cost. \textit{UnlinkableDFL} preserves learning utility as the deployment scales, drives a passive network-layer linking attack to the random baseline under the full design, and exposes the expected latency and traffic overheads of mixing. Its scope is deliberately bounded: once a recipient decrypts fragments, model content can still carry payload-layer fingerprints. Network-layer unlinkability should therefore complement model-layer defenses such as Shatter and DivShare. Future work should reduce mixnet latency, tune routing and rate parameters to workload and topology, and combine the framework with payload-layer obfuscation.

% the acknowledgment should be added when CRC
% \section*{Acknowledgment}
% The preferred spelling of the word ``acknowledgment'' in America is without 
% an ``e'' after the ``g''. Avoid the stilted expression ``one of us (R. B. 
% G.) thanks $\ldots$''. Instead, try ``R. B. G. thanks$\ldots$''. Put sponsor 
% acknowledgments in the unnumbered footnote on the first page.

\bibliographystyle{ieeetr}
\bibliography{reference}

\noindent \small{\\All links above were last accessed on \today.}

\clearpage
\onecolumn
\appendices
\section*{Supplementary Material}
\setcounter{theorem}{0}
\setcounter{lemma}{0}
\setcounter{assumption}{0}
\setcounter{remark}{0}
\setcounter{figure}{0}
\setcounter{table}{0}
\renewcommand{\thefigure}{\thesection.\arabic{figure}}
\renewcommand{\thetable}{\thesection.\arabic{table}}

\section*{Overview}

This supplementary material is organized into three categories. Appendix~\ref{supp:proofs} collects the full proofs for the theoretical claims stated in the main paper, including network-layer unlinkability, Single-Use Reply Block (SURB) acknowledgment independence, \texttt{FragFedAvg} learning preservation, and \texttt{FragKrum} robustness. Appendix~\ref{supp:implementation} summarizes prototype configuration, experiment orchestration, and interface screenshots. Appendix~\ref{supp:experiments} provides extended experimental settings and supporting results that complement the main evaluation.

The proofs are conditional on the threat model in the main paper. In particular, the network-layer results concern observations available from packet timing, routing, compromised relay views, and visible packet metadata. They do not imply differential privacy for the learned model and do not remove content-based fingerprints available after a destination decrypts a fragment.

\section{Proofs}
\label{supp:proofs}

\subsection{Network-Layer Unlinkability}
\label{supp:network}

\subsubsection{Notation}

Let \(m\) be a target fragment packet, and let \(S(m)\) be its true sender. Let \(T\) denote the adversary's network-layer trace. The trace may include packet times, visible transport links, observations from compromised nodes, and all public system parameters. Let \(\Pi_T(m)\) be the set of sender-labeled trajectories for \(m\) that remain feasible after conditioning on \(T\). A trajectory \(\pi\in\Pi_T(m)\) specifies a candidate sender, a route, and the input-output choices through relay shuffles on that route. Let \(h(\pi)\) denote the number of hidden forwarding shuffles on \(\pi\).

For a candidate sender \(u\), define
\begin{equation}
    \Gamma_T(u,m)=\{\pi\in\Pi_T(m): \pi \text{ starts at } u\}.
\end{equation}
Let
\begin{equation}
    \gamma_T=\max_{u\in\mathcal{U}_T}|\Gamma_T(u,m)|,
\end{equation}
where \(\mathcal{U}_T\) is the set of candidate senders consistent with \(T\). The conditional sender-identification probability is
\begin{equation}
    p_{\mathrm{link}}(m\mid T)
    =\max_{u\in\mathcal{U}_T}\Pr[S(m)=u\mid T].
\end{equation}
The advantage \(\mathbb{E}_T[p_{\mathrm{link}}(m\mid T)-1/|\mathcal{U}_T|]\) is non-negative since, for each trace, the maximum posterior over \(\mathcal{U}_T\) is at least \(1/|\mathcal{U}_T|\). The conditional trajectory min-entropy is
\begin{equation}
H_{\infty}(\Pi_m \mid T)
= -\log_2 \max_{\pi\in\Pi_T(m)}\Pr[\Pi_m=\pi\mid T].
\end{equation}

Table~\ref{tab:supp-notation} summarizes the notation used across the proofs.

\begin{table}[H]
\centering
\caption{Notation used in the supplementary proofs.}
\label{tab:supp-notation}
\small
\renewcommand{\arraystretch}{1.08}
\begin{tabularx}{\linewidth}{@{}p{0.23\linewidth}X@{}}
\toprule
\textbf{Symbol} & \textbf{Meaning} \\
\midrule
\multicolumn{2}{@{}l}{\textit{Network-layer unlinkability}} \\
\(m\) & Target model-fragment packet. \\
\(S(m)\) & True sender of packet \(m\). \\
\(T\) & Adversary's network-layer trace, including timing, visible links, compromised relay views, and public parameters. \\
\(\pi\), \(\Pi_T(m)\) & A sender-labeled trajectory and the set of trajectories for \(m\) that remain feasible after conditioning on \(T\). \\
\(u\), \(\mathcal{U}_T\) & Candidate sender and the trace-consistent candidate sender set. \\
\(\Gamma_T(u,m)\) & Feasible trajectories that attribute packet \(m\) to candidate sender \(u\). \\
\(\gamma_T\) & Maximum number of feasible trajectories assigned to any one candidate sender. \\
\(R_K\), \(K\) & Random route sampled under maximum path length \(K\). \\
\(h\), \(O\) & Minimum number of hidden forwarding shuffles over feasible trajectories and configured outbox size. \\
\(\Pi_f\), \(\Pi_a\) & Forward fragment trajectory and acknowledgment trajectory encoded in the SURB. \\
\(T_0\), \(T_a\) & Trace before acknowledgment emission and visible acknowledgment trace. \\
\midrule
\multicolumn{2}{@{}l}{\textit{Learning preservation}} \\
\(F\), \(F_k\) & Global objective and local objective at peer \(k\). \\
\(w_t\), \(\tilde w_t\) & Model state before aggregation and fragment-based aggregate at round \(t\). \\
\(L\), \(\eta\), \(\tau\) & Smoothness constant, SGD step size, and number of local SGD steps. \\
\(\sigma_g^2\), \(\nu_q\) & Stochastic-gradient variance bound and quantization variance bound. \\
\(Q\) & Unbiased stochastic quantizer applied before lossless compression. \\
\(P_t\) & Fragment pool available at a node in round \(t\). \\
\(c_{t,i}\), \(p_{t,i}\), \(\bar p\) & Number of fragments covering index \(i\), its coverage probability, and the uniform lower bound on coverage. \\
\(\bar K_{\mathrm{eff}}\) & Minimum expected number of contributors among covered indices. \\
\(\Delta_{\mathrm{fb}}\) & Bounded optimization error introduced by fallback on uncovered indices. \\
\midrule
\multicolumn{2}{@{}l}{\textit{Fragment-level robustness}} \\
\(Y_r\), \(n_r\), \(B_r\) & Candidate slices for index range \(r\), number of candidates, and upper bound on Byzantine slices. \\
\(q_r\), \(N_y\), \(s_y\) & Number of nearest neighbors used by \texttt{FragKrum}, the neighbor set of candidate \(y\), and its Krum score. \\
\(H_r\), \(\delta_r\) & Honest slices for range \(r\) and their diameter bound. \\
\bottomrule
\end{tabularx}
\end{table}
\begin{theorem}[Network-layer sender--message unlinkability]
\label{supp:thm:network}
Consider a fragment packet \(m\) whose sender-labeled trajectories \(\Pi_T(m)\) are consistent with the adversary trace \(T\). Let \(h\) be the minimum number of hidden forwarding shuffles on any feasible trajectory. That is, \(h=\min_{\pi\in\Pi_T(m)}h(\pi)\). Assume that visible packet features are sender-independent and that each hidden forwarding shuffle uses a uniformly shuffled outbox of size \(O\). If the adversary uses only network-layer observations, its best sender-linking probability is bounded by
\begin{equation}
\label{supp:eq:sender-link-bound}
p_{\mathrm{link}}(m \mid T)
\le
\min\left\{1,\gamma_T 2^{-H_{\infty}(\Pi_m\mid T)}\right\}.
\end{equation}
Moreover, if the hidden shuffle choices remain unrevealed except through their feasible outbox positions,
\begin{equation}
\label{supp:eq:path-min-entropy}
H_{\infty}(\Pi_m\mid T)
\ge
H_{\infty}(R_K\mid T)+h\log_2 O,
\end{equation}
where \(R_K\) denotes the random route sampled under maximum path length \(K\), and \(H_{\infty}(R_K\mid T)\) is the remaining route uncertainty after observing \(T\). Combining Eqs.~\ref{supp:eq:sender-link-bound} and~\ref{supp:eq:path-min-entropy} gives
\begin{equation}
\label{supp:eq:explicit-ko-bound}
p_{\mathrm{link}}(m \mid T)
\le
\min\left\{1,
\gamma_T 2^{-H_{\infty}(R_K\mid T)} O^{-h}\right\},
\end{equation}
where \(h\le K\). The relay-shuffle contribution is at most \(K\log_2 O\), reached only when the least-hidden feasible trajectory still hides all \(K\) forwarding stages.
\end{theorem}

\begin{proof}
The proof has two parts. The first part converts trajectory min-entropy into a sender-linking bound. The second part lower-bounds the trajectory min-entropy contributed by route choice and hidden relay shuffles.

Fix a candidate sender \(u\in\mathcal{U}_T\). The posterior probability that \(u\) is the sender equals the sum of the posterior probabilities of all feasible trajectories that start at \(u\):
\begin{equation}
\Pr[S(m)=u\mid T]
=\sum_{\pi\in\Gamma_T(u,m)}\Pr[\Pi_m=\pi\mid T].
\end{equation}
By the definition of min-entropy, every feasible trajectory satisfies
\begin{equation}
\Pr[\Pi_m=\pi\mid T]\le 2^{-H_{\infty}(\Pi_m\mid T)}.
\end{equation}
Therefore,
\begin{equation}
\Pr[S(m)=u\mid T]
\le
|\Gamma_T(u,m)|2^{-H_{\infty}(\Pi_m\mid T)}
\le
\gamma_T2^{-H_{\infty}(\Pi_m\mid T)}.
\end{equation}
The adversary's optimal sender-linking strategy is to choose the sender with maximum posterior probability, so
\begin{equation}
p_{\mathrm{link}}(m\mid T)
=\max_{u\in\mathcal{U}_T}\Pr[S(m)=u\mid T]
\le
\gamma_T2^{-H_{\infty}(\Pi_m\mid T)}.
\end{equation}
Since a probability is at most one, Eq.~\ref{supp:eq:sender-link-bound} follows.

We now prove Eq.~\ref{supp:eq:path-min-entropy}. Let a feasible sender-labeled trajectory be represented as
\begin{equation}
    \Pi_m=(U,R,C_1,\ldots,C_{h_\pi}),
\end{equation}
where \(U\) is the sender label, \(R\) is the sampled overlay route, and \(C_\ell\) is the local input-output matching choice through the \(\ell\)-th hidden shuffle on that route. By definition, \(h_\pi\ge h\) for every feasible trajectory. Compromised relays and monitored links are already included in the conditioning trace \(T\), and therefore do not contribute hidden entropy. Only hidden shuffles contribute the \(C_\ell\) terms.

For any route value \(r\), the definition of route min-entropy gives
\begin{equation}
\Pr[R=r\mid T]\le 2^{-H_{\infty}(R_K\mid T)}.
\end{equation}
At each hidden relay shuffle, the relay draws an outbox of size \(O\) and applies a uniformly sampled permutation before emission. Conditioned on the packet being in that outbox and on all information in \(T\), the adversary cannot distinguish which of the \(O\) outgoing positions corresponds to the incoming packet. Hence
\begin{equation}
\Pr[C_\ell=c_\ell\mid U=u,R=r,T,C_1=c_1,\ldots,C_{\ell-1}=c_{\ell-1}]
\le \frac{1}{O}.
\end{equation}
Using \(\Pr[U=u,R=r\mid T]\le\Pr[R=r\mid T]\), for any feasible trajectory \(\pi=(u,r,c_1,\ldots,c_{h_\pi})\),
\begin{align}
\Pr[\Pi_m=\pi\mid T]
&=\Pr[U=u,R=r\mid T]
  \prod_{\ell=1}^{h_\pi}
  \Pr[C_\ell=c_\ell\mid U=u,R=r,T,C_{<\ell}=c_{<\ell}] \\
&\le
\Pr[R=r\mid T] O^{-h_\pi} \\
&\le
2^{-H_{\infty}(R_K\mid T)} O^{-h}.
\end{align}
Taking the maximum over feasible trajectories and applying \(-\log_2(\cdot)\) gives
\begin{equation}
H_{\infty}(\Pi_m\mid T)
\ge H_{\infty}(R_K\mid T)+h\log_2 O.
\end{equation}
This proves Eq.~\ref{supp:eq:path-min-entropy}. Substituting it into Eq.~\ref{supp:eq:sender-link-bound} gives Eq.~\ref{supp:eq:explicit-ko-bound}. Since a route has at most \(K\) forwarding stages, \(h\le K\), so the relay-shuffle contribution is at most \(K\log_2 O\). This maximum is reached only when the least-hidden feasible trajectory still hides all \(K\) forwarding stages. This proves the theorem.
\end{proof}

\subsubsection{Compromised-Relay Regime}

The compromised-relay corollary in the main paper is an idealized calculation for interpreting the effect of partial relay compromise. If each forwarding stage samples a relay independently and uniformly from a population where a fraction \(c\) is compromised, a length-\(K\) trajectory has all forwarding stages compromised with probability \(c^K\). With probability \(1-c^K\), at least one forwarding shuffle remains hidden from the adversary, and the sender-linking bound keeps the additional \(O^{-1}\) relay-shuffle factor from Theorem~\ref{supp:thm:network}. The expected number of hidden shuffles under the same sampling model is \(K(1-c)\).

For topology-constrained deployments, the same reasoning applies after replacing \(c^K\) with the actual probability that every forwarding stage on the sampled route is compromised. This is why the main paper states the corollary as a regime calculation rather than as a topology-independent guarantee.
\begin{remark}[Role of packet indistinguishability]
The proof assumes that the adversary cannot separate real fragments, cover packets, relay packets, and acknowledgments by visible packet features. In \textit{UnlinkableDFL}, this follows from fixed-size packet formatting, padding, and onion encryption. If packet classes became externally distinguishable, the feasible trajectory set \(\Pi_T(m)\) would shrink and the bound would weaken accordingly.
\end{remark}

\subsection{Acknowledgment Independence}
\label{supp:surb}

\begin{lemma}[SURB acknowledgment independence]
\label{supp:lemma:surb}
Let \(\Pi_f\) be the forward trajectory of a fragment and let \(\Pi_a\) be the acknowledgment trajectory encoded in its Single-Use Reply Block (SURB). Let \(T_0\) include the trace before acknowledgment emission and the event that the destination generated an acknowledgment. Suppose that \(\Pi_a\) is sampled independently of \(\Pi_f\) conditioned on \(T_0\), that the destination uses the SURB without learning the sender address, and that the acknowledgment packet has the same externally visible packet format as other mixnet packets. Then observing the visible acknowledgment trace does not reduce the posterior uncertainty of the forward trajectory except through information already present in \(T_0\).
\end{lemma}

\begin{proof}
Let \(T_a\) denote the visible acknowledgment trace. By construction, the SURB contains encrypted routing instructions for the return path and does not reveal the sender address to the destination. The return trajectory \(\Pi_a\) is sampled independently of the forward trajectory \(\Pi_f\) conditioned on \(T_0\). Since acknowledgment packets are formatted and routed as ordinary mixnet packets, \(T_a\) is a function of \(\Pi_a\) and trace-visible randomness, but not of \(\Pi_f\) beyond information already included in \(T_0\). Thus \(T_a\) is conditionally independent of \(\Pi_f\) given \(T_0\), and therefore
\begin{equation}
    \Pr[\Pi_f=\pi_f\mid T_0,T_a]
    =\Pr[\Pi_f=\pi_f\mid T_0].
\end{equation}
Consequently,
\begin{equation}
    H_{\infty}(\Pi_f\mid T_0,T_a)
    =H_{\infty}(\Pi_f\mid T_0).
\end{equation}
The acknowledgment trace creates a separate inference problem for \(\Pi_a\), but it does not collapse the forward path to a direct destination-to-sender relation.
\end{proof}

\subsection{Learning Preservation of \texttt{FragFedAvg}}
\label{supp:fragfedavg}

\subsubsection{Assumptions}

Let \(F(w)=\frac{1}{N}\sum_{k=1}^N F_k(w)\). We use the following standard assumptions for FedAvg-style non-convex analysis.

\begin{assumption}[Smoothness and lower boundedness]
Each local objective \(F_k\) is \(L\)-smooth, and the global objective \(F\) is bounded below by \(F^\star\).
\end{assumption}

\begin{assumption}[Stochastic gradients]
Each stochastic gradient is unbiased and has bounded variance \(\sigma_g^2\):
\begin{equation}
    \mathbb{E}[g_k(w)]=\nabla F_k(w),
    \qquad
    \mathbb{E}\|g_k(w)-\nabla F_k(w)\|^2\le \sigma_g^2.
\end{equation}
\end{assumption}

\begin{assumption}[Base FedAvg descent rate]
\label{supp:asm:base-rate}
The usual FedAvg local-drift terms are bounded under \(\tau\) local SGD steps and step size \(\eta\). Concretely, we take as given the standard non-convex result for FedAvg/FedProx-style methods: when all participating model states are aggregated by index-wise averaging over an effective participating set of size \(K_{\mathrm{eff}}\), the iterates satisfy a per-round descent inequality of the form
\begin{equation}
\label{supp:eq:base-descent}
\mathbb{E}[F(\bar w_{t+1})]
\le
\mathbb{E}[F(\bar w_t)]
-c\eta\,\mathbb{E}\|\nabla F(\bar w_t)\|^2
+C\eta^2 L\sigma_g^2
+\frac{C\eta}{K_{\mathrm{eff}}\tau},
\end{equation}
for universal constants \(c,C>0\) and \(\eta\le 1/(2L)\)~\cite{mcmahan2017communication,karimireddy2020scaffold}. Our analysis treats \texttt{FragFedAvg} as a perturbation of this baseline rather than re-deriving it.
\end{assumption}

\begin{assumption}[Unbiased quantization]
The fragment quantizer \(Q\) is unbiased and has bounded variance \(\nu_q\):
\begin{equation}
    \mathbb{E}[Q(x)\mid x]=x,
    \qquad
    \mathbb{E}\|Q(x)-x\|^2\le \nu_q.
\end{equation}
The compression stage is lossless.
\end{assumption}

\begin{assumption}[Coverage and fallback]
For parameter index \(i\) in round \(t\), let \(c_{t,i}\) be the number of available fragments covering \(i\), and let \(p_{t,i}=\Pr(c_{t,i}\ge 1)\). Define \(\bar p=\min_i\inf_t p_{t,i}\). Among covered indices, the expected number of contributors is at least \(\bar K_{\mathrm{eff}}\). When an index is uncovered and the algorithm falls back to the pre-aggregation local value, the induced optimization error is bounded by \(\Delta_{\mathrm{fb}}\).
\end{assumption}

\begin{theorem}[Learning preservation of \texttt{FragFedAvg}]
\label{supp:thm:fragfedavg}
Under the assumptions above and step size \(\eta\le 1/(2L)\), the fragment-based iterate \(\tilde w_t\) (with \(w_{t+1}:=\tilde w_t\)) satisfies the bound
\begin{equation}
\label{supp:eq:frag-rate}
\begin{split}
    \frac{1}{T}\sum_{t=0}^{T-1}
    \mathbb{E}\|\nabla F(\tilde w_t)\|^2
    \le
    O\Big(
        \frac{1}{\eta T}
        + \eta L(\sigma_g^2+\nu_q)
        + \frac{1}{\bar p\,\bar K_{\mathrm{eff}}\tau}
        {} \\
        \qquad
        + (1-\bar p)\Delta_{\mathrm{fb}}
    \Big).
\end{split}
\end{equation}
The first three terms vanish as \(T\to\infty\) and \(\eta\to0\). The last term is a coverage-dependent bias floor that vanishes only as \(\bar p\to1\). The statement is therefore a preservation result: \texttt{FragFedAvg} matches the FedAvg stationarity rate up to a floor controlled by fragment coverage and quantization, and reduces exactly to FedAvg when \(\bar p=1\) and \(\nu_q=0\).
\end{theorem}

\begin{proof}
The proof is a reduction: we treat \texttt{FragFedAvg} as the FedAvg baseline of Assumption~\ref{supp:asm:base-rate} subject to three perturbations --- incomplete coverage, quantization noise, and fallback bias --- and propagate each through the base descent inequality~\eqref{supp:eq:base-descent}.

We fix the iterate convention \(w_{t+1}:=\tilde w_t\), i.e., each round replaces the local model with the fragment-based aggregate. Let \(\bar w_t\) denote the FedAvg-style aggregate that would be obtained from the same contributors if all their model states were available in unquantized form, and let \(\tilde w_t\) denote the \texttt{FragFedAvg} aggregate. For a covered index \(i\), \texttt{FragFedAvg} computes
\begin{equation}
    \tilde w_{t,i}
    =\frac{1}{c_{t,i}}\sum_{f\in P_t: i\in f}Q(f[i]).
\end{equation}
Because \(Q\) is unbiased,
\begin{equation}
    \mathbb{E}[\tilde w_{t,i}\mid P_t,c_{t,i}>0]
    =\frac{1}{c_{t,i}}\sum_{f\in P_t:i\in f}f[i].
\end{equation}
Thus quantization does not introduce bias on covered indices. Its conditional variance is bounded by
\begin{equation}
    \mathbb{E}\left[\left\|\tilde w_{t,i}
    -\frac{1}{c_{t,i}}\sum_{f\in P_t:i\in f}f[i]\right\|^2
    \middle| P_t,c_{t,i}>0\right]
    \le \frac{\nu_q}{c_{t,i}}.
\end{equation}
Taking expectation over covered indices and using the lower bound \(\bar K_{\mathrm{eff}}\) gives a quantization contribution of order \(\nu_q/\bar K_{\mathrm{eff}}\). Since \(\bar K_{\mathrm{eff}}\ge 1\), this is at most \(\nu_q\), which we use for a coverage-independent statement, so in the descent bound it appears with the same smoothness multiplier as stochastic-gradient variance, yielding the \(\eta L\nu_q\) term. A tighter \(\nu_q/\bar K_{\mathrm{eff}}\) form is available when coverage is high.

Coverage affects the number of contributors used for each index. Since each index is covered with probability at least \(\bar p\), and covered indices have at least \(\bar K_{\mathrm{eff}}\) expected contributors, the effective participation term is degraded from the usual FedAvg participation factor to
\begin{equation}
    \frac{1}{\bar p\,\bar K_{\mathrm{eff}}\tau}.
\end{equation}
This is the same role played by the participation term in standard FedAvg analyses, but with fragment coverage included.

When \(c_{t,i}=0\), \texttt{FragFedAvg} uses the pre-aggregation local value for that index. By assumption, the induced optimization error is bounded by \(\Delta_{\mathrm{fb}}\). Since this event has probability at most \(1-\bar p\), the fallback contribution is bounded by
\begin{equation}
    (1-\bar p)\Delta_{\mathrm{fb}}.
\end{equation}

It remains to combine these terms with the base descent inequality. Assumption~\ref{supp:asm:base-rate} provides the descent~\eqref{supp:eq:base-descent} for the idealized aggregate \(\bar w_t\). The \texttt{FragFedAvg} iterate differs from \(\bar w_t\) only through the three perturbations bounded above, so \(\mathbb{E}\|\tilde w_t-\bar w_t\|^2\) is controlled by the quantization variance (order \(\nu_q\)) and the fallback deviation (order \((1-\bar p)\Delta_{\mathrm{fb}}\)). Because \(F\) is \(L\)-smooth, the gradient measured at the realized iterate is related to the one in~\eqref{supp:eq:base-descent} by
\begin{equation}
\label{supp:eq:grad-bridge}
\|\nabla F(\bar w_t)-\nabla F(\tilde w_t)\|
\le L\,\|\bar w_t-\tilde w_t\|,
\end{equation}
so replacing \(\nabla F(\bar w_t)\) by \(\nabla F(\tilde w_t)\) costs only an additional term of the same order as the perturbations already accounted for (via Young's inequality on \(\|\nabla F(\bar w_t)\|^2\)). Substituting the local SGD update into~\eqref{supp:eq:base-descent}, taking expectation, using unbiased stochastic gradients and the local-drift bound, applying the bridge~\eqref{supp:eq:grad-bridge}, and then adding the quantization, coverage, and fallback perturbations derived above yields
\begin{equation}
\mathbb{E}[F(w_{t+1})]
\le
\mathbb{E}[F(w_t)]
-c\eta\mathbb{E}\|\nabla F(\tilde w_t)\|^2
+C\eta^2 L(\sigma_g^2+\nu_q)
+C\eta\left(\frac{1}{\bar p\bar K_{\mathrm{eff}}\tau}+(1-\bar p)\Delta_{\mathrm{fb}}\right),
\end{equation}
for universal constants \(c,C>0\) under \(\eta\le 1/(2L)\). Summing over \(t=0,\ldots,T-1\), telescoping the left-hand side, and using \(F(w_T)\ge F^\star\) gives
\begin{equation}
\frac{1}{T}\sum_{t=0}^{T-1}\mathbb{E}\|\nabla F(\tilde w_t)\|^2
\le
O\left(
\frac{F(w_0)-F^\star}{\eta T}
+\eta L(\sigma_g^2+\nu_q)
+\frac{1}{\bar p\bar K_{\mathrm{eff}}\tau}
+(1-\bar p)\Delta_{\mathrm{fb}}
\right).
\end{equation}
Absorbing \(F(w_0)-F^\star\) into the big-\(O\) notation proves Eq.~\ref{supp:eq:frag-rate}.
\end{proof}

\begin{remark}
When \(\bar p=1\) and \(\nu_q=0\), every index is covered and no quantization noise is introduced. In this case \texttt{FragFedAvg} reduces to FedAvg over the same participating model states, and the additional coverage, fallback, and quantization terms vanish.
\end{remark}

\subsection{Robustness of \texttt{FragKrum}}
\label{supp:fragkrum}

Let \(Y_r\) be the multiset of candidate fragment slices for index range \(r\). Let \(n_r=|Y_r|\), and suppose at most \(B_r\) slices in \(Y_r\) are Byzantine. For a candidate \(y\in Y_r\), \texttt{FragKrum} uses
\begin{equation}
    q_r=n_r-B_r-2
\end{equation}
and defines \(N_y\) as the \(q_r\) nearest slices to \(y\) in \(Y_r\setminus\{y\}\). Its score is
\begin{equation}
    s_y=\sum_{z\in N_y}\|y-z\|_2^2.
\end{equation}
The selected slice is \(y^\star=\arg\min_{y\in Y_r}s_y\).

\begin{theorem}[Fragment-level Krum robustness]
\label{supp:thm:fragkrum}
For an index range \(r\), assume \(Y_r\) contains \(n_r\) candidate slices, at most \(B_r\) of them are Byzantine, and \(n_r>2B_r+2\). Let \(H_r\subseteq Y_r\) be the honest slices, and assume the honest slices have diameter at most \(\delta_r\):
\begin{equation}
    \|y-z\|_2\le \delta_r
    \qquad
    \forall y,z\in H_r.
\end{equation}
Then \texttt{FragKrum} selects a slice \(y^\star\) satisfying
\begin{equation}
    s_{y^\star}
    \le (n_r-B_r-2)\delta_r^2.
\end{equation}
Moreover, \(y^\star\) is close to the honest cluster in the sense that
\begin{equation}
    \min_{h\in H_r}\|y^\star-h\|_2^2
    \le
    \frac{n_r-B_r-2}{n_r-2B_r-2}\delta_r^2.
\end{equation}
The rule therefore inherits the standard Krum requirement that each range have more than \(2B_r+2\) candidates, but it does not require sender identities.
\end{theorem}

\begin{proof}
Let \(q_r=n_r-B_r-2\). Since \(n_r>2B_r+2\), we have \(q_r>B_r\) and \(q_r-B_r=n_r-2B_r-2>0\).

Consider any honest slice \(h\in H_r\). The number of honest slices other than \(h\) is
\begin{equation}
    |H_r|-1=(n_r-B_r)-1=n_r-B_r-1.
\end{equation}
Since \(q_r=n_r-B_r-2\), there are at least \(q_r\) honest neighbors available for \(h\). All honest slices are within distance \(\delta_r\) of one another, so the sum of squared distances from \(h\) to its \(q_r\) nearest neighbors is at most
\begin{equation}
    s_h\le q_r\delta_r^2.
\end{equation}
Because \(y^\star\) minimizes the Krum score over all candidates,
\begin{equation}
    s_{y^\star}\le s_h\le q_r\delta_r^2
    =(n_r-B_r-2)\delta_r^2.
\end{equation}
This proves the score bound.

Now inspect the \(q_r\) nearest neighbors of \(y^\star\). At most \(B_r\) candidates in all of \(Y_r\) are Byzantine, so at least
\begin{equation}
    q_r-B_r=n_r-2B_r-2
\end{equation}
honest slices appear among these nearest neighbors. The score \(s_{y^\star}\) is the sum of squared distances to all \(q_r\) nearest neighbors, including those honest ones. Therefore,
\begin{equation}
    \sum_{h\in N_{y^\star}\cap H_r}\|y^\star-h\|_2^2
    \le s_{y^\star}
    \le q_r\delta_r^2.
\end{equation}
Since there are at least \(q_r-B_r\) honest slices in \(N_{y^\star}\), the smallest squared distance from \(y^\star\) to an honest slice is at most the average over these honest neighbors:
\begin{equation}
    \min_{h\in H_r}\|y^\star-h\|_2^2
    \le
    \frac{q_r}{q_r-B_r}\delta_r^2
    =
    \frac{n_r-B_r-2}{n_r-2B_r-2}\delta_r^2.
\end{equation}
Thus the selected slice must lie near the honest cluster whenever the range-wise Krum condition holds. The computation uses only distances between slices in the same index range and never uses sender labels, which establishes the stated sender-independent robustness property.
\end{proof}

\begin{remark}[Partial coverage]
The theorem is range-local. If a range has too few candidate slices or too many Byzantine slices, the Krum condition fails for that range. This is why the main paper treats \texttt{FragKrum} as robust only when each covered range has enough honest candidates, and why Byzantine resilience is evaluated empirically under varying attack intensity.
\end{remark}

\FloatBarrier
\clearpage
\section{Implementation Details and Experiment Interface}
\label{supp:implementation}
\setcounter{figure}{0}
\setcounter{table}{0}
\renewcommand{\thefigure}{\thesection.\arabic{figure}}
\renewcommand{\thetable}{\thesection.\arabic{table}}

This appendix gives the implementation details needed to interpret the prototype and reproduce the experimental workflow in the main paper. The key distinction is between the decentralized training path and the experiment-management layer. During a run, each node performs local training, route sampling, mixing, acknowledgment handling, fragment storage, and aggregation by itself. The manager only prepares the nodes, injects configuration, schedules controlled scenarios, and collects logs.

\subsection{Configuration Surface}

The prototype uses a typed configuration object, \texttt{FullNodeConfig}, that is serialized by the experiment manager and injected into each node container at startup. This keeps node-side code fixed while allowing experiments to vary learning, topology, transport, payload, aggregation, and stressor parameters. Table~\ref{tab:supp_config_groups} lists the groups used by the experiments in Appendix~\ref{supp:experiments}.

\begin{table}[H]
\centering
\caption{Prototype configuration groups used to instantiate experiments.}
\label{tab:supp_config_groups}
\small
\begin{tabularx}{\linewidth}{@{}p{0.24\linewidth}X@{}}
\toprule
\textbf{Group} & \textbf{Configured values} \\
\midrule
Learning task & dataset, model architecture, local epochs, rounds, batch size, optimizer settings \\
Topology & node count, graph family, graph degree, join schedule, exit schedule \\
Mixnet transport & maximum path length \(K\), outbox size \(O\), delay parameters \(\mu,\sigma\), cover-packet policy \\
Payload handling & Sphinx body size, fragment size, quantization bit width, lossless compression, partial-update ratio \\
Aggregation and attacks & \texttt{FragFedAvg}, \texttt{FragKrum}, Byzantine budget, attack type, attack magnitude \\
Monitoring & communication, delivery, routing, learning, CPU, and memory metrics \\
\bottomrule
\end{tabularx}
\end{table}

\subsection{Experiment Orchestration Boundary}

The FastAPI manager provides repeatable deployment rather than a training coordinator. Before a run starts, it generates per-scenario Sphinx keys and Transport Layer Security (TLS) certificates, builds the topology, validates configuration constraints, and launches each node either as a Docker container or as a core-pinned OS process. During training, it does not aggregate models, choose mixnet routes, select fragments, or act as a registration authority. This boundary matches the clarification in the main paper: the learning and anonymous communication workflow is decentralized, while controlled experiments still need tooling for setup, observation, and teardown.

The manager also schedules controlled churn and attack events. For example, a late-join experiment starts one node after the initial rounds, and an exit experiment marks one node inactive after a specified round. These events are injected for reproducibility, but the response to them is handled by the nodes through local peer views, fresh route sampling, retransmission, fragment-pool aggregation, and fallback for uncovered parameter ranges.

\subsection{Dashboard and Logged Metrics}

The dashboard is included here to document what was configurable and what was measured, not as a privacy-critical component. Figure~\ref{fig:supp_ui_config} shows the configuration view used before launch. Figure~\ref{fig:supp_ui_indicators} shows pre-run indicators derived from the same scenario model, including topology diameter, reachable peers, fragments per model, estimated bandwidth, entropy diagnostics, and \texttt{FragKrum} feasibility. Figure~\ref{fig:supp_ui_monitoring} shows the live monitoring view used during a run. Metrics are streamed through Server-Sent Events (SSE) and exported to comma-separated value (CSV) logs after completion.

\begin{figure}[htbp]
\centering
\includegraphics[width=1\linewidth]{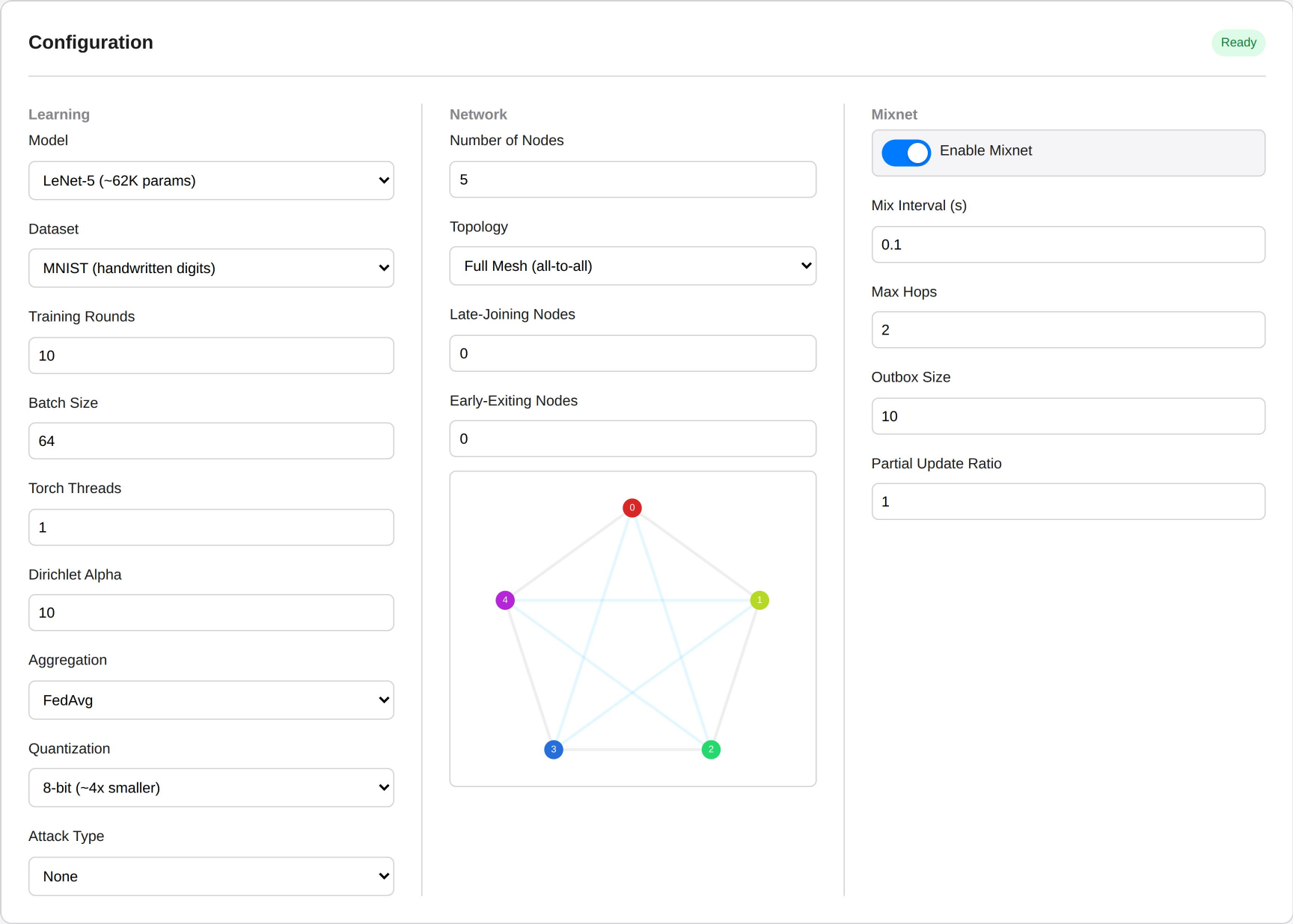}
\caption{Configuration interface for setting learning, topology, mixnet, payload, and stressor parameters before container launch.}
\label{fig:supp_ui_config}
\end{figure}

\begin{figure}[htbp]
\centering
\includegraphics[width=1\linewidth]{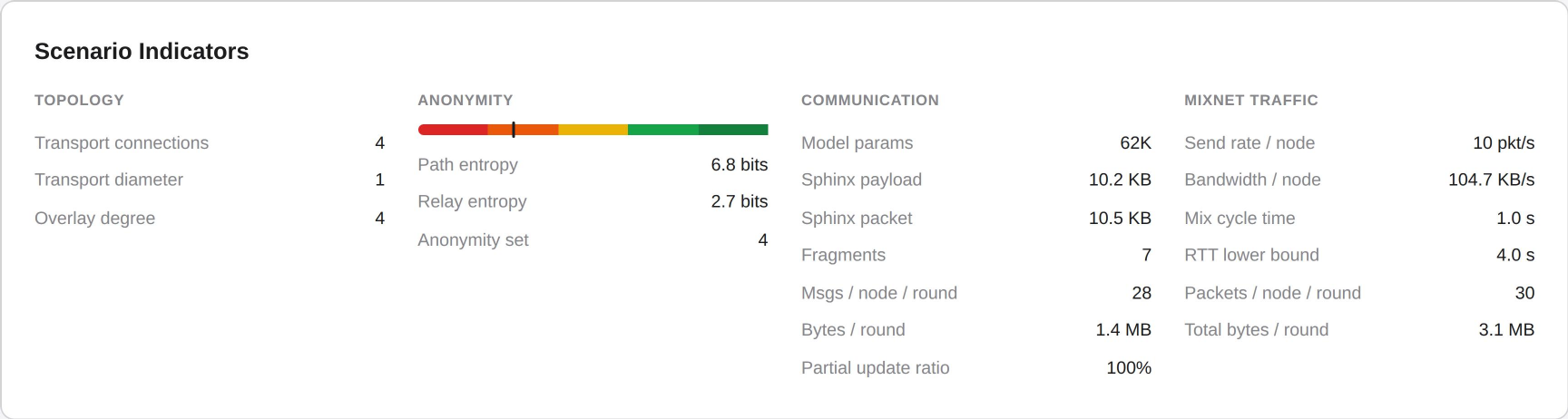}
\caption{Pre-run indicator panel. The values help validate whether a configuration has the intended reachability, payload volume, entropy diagnostics, and robustness budget before execution.}
\label{fig:supp_ui_indicators}
\end{figure}

\begin{figure}[htbp]
\centering
\includegraphics[width=1\linewidth]{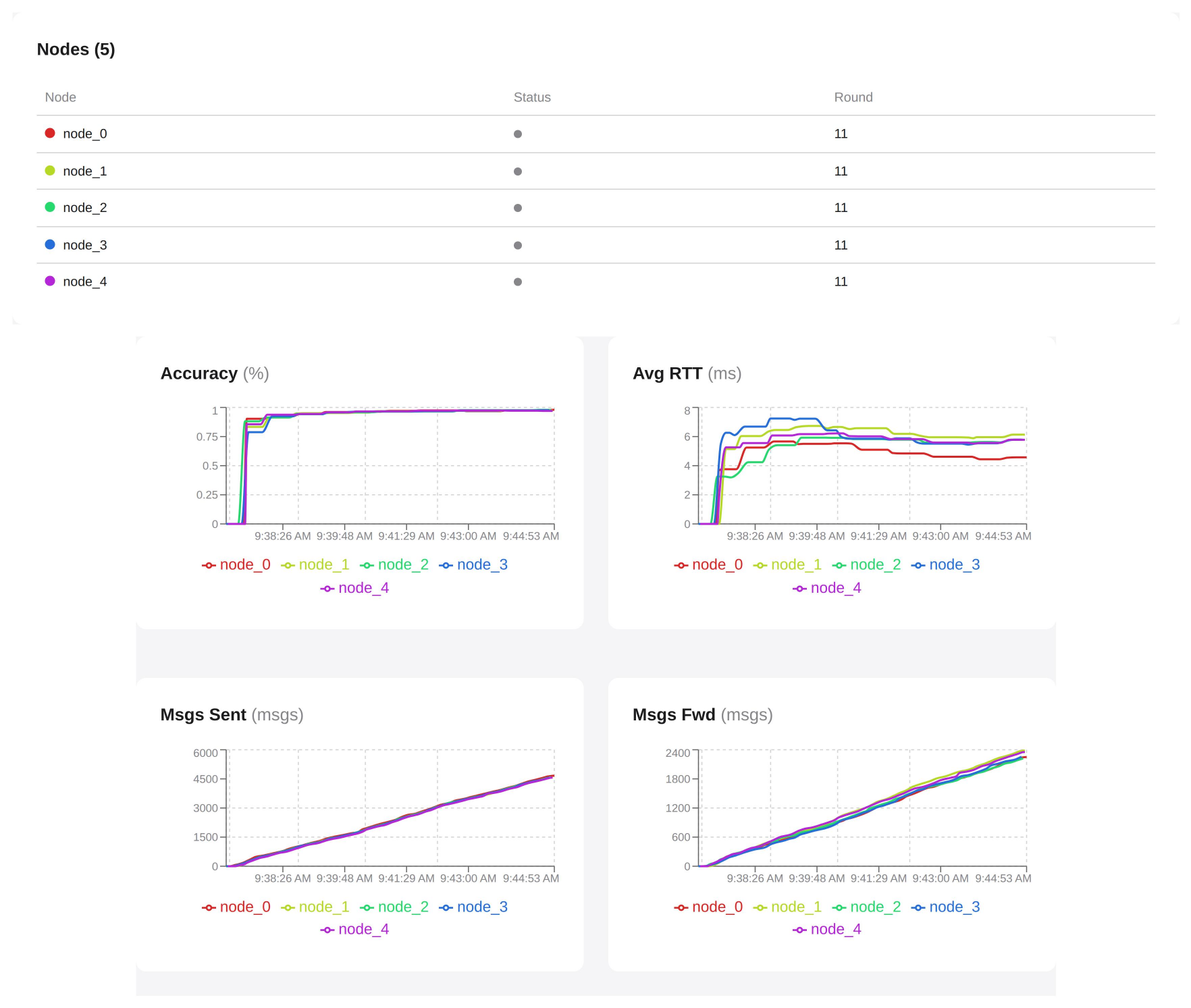}
\caption{Monitoring dashboard used during a run to inspect node state, fragment exchange, learning progress, and resource usage.}
\label{fig:supp_ui_monitoring}
\end{figure}
\FloatBarrier
\section{Supplementary Experiments}
\label{supp:experiments}
\setcounter{figure}{0}
\setcounter{table}{0}
\renewcommand{\thefigure}{\thesection.\arabic{figure}}
\renewcommand{\thetable}{\thesection.\arabic{table}}

This appendix reports the experiments that support the evaluation section of the main paper. The goal is not to introduce additional claims, but to document the settings, full curves, and secondary diagnostics behind the selected main-paper figures. Unless stated otherwise, experiments use the default prototype setting from the main paper: \(N=32\), degree-4 circulant topology, maximum path length \(K=2\), outbox size \(O=100\), delay mean \(\mu=0.1\)~s, delay standard deviation \(\sigma=0.001\), LeNet-5, ten training rounds, Dirichlet partitioning with \(\alpha=10\), 8-bit stochastic quantization with lossless compression, and \texttt{FragFedAvg}.

\begin{table}[H]
\centering
\caption{Supplementary experiment groups and their role in the main evaluation.}
\label{tab:supp_exp_settings}
\small
\begin{tabularx}{\linewidth}{@{}p{0.24\linewidth}X@{}}
\toprule
\textbf{Group} & \textbf{Contents and role} \\
\midrule
C.1 Learning/payload & Convergence by \(N\), comparison with Fedstellar, fragment counts, bytes per round, partial-update convergence, and data heterogeneity. \\
C.2 Attack/entropy/cost & Trace fields, adversary channels, attack procedure, entropy and topology diagnostics, mix-on/off overhead, cover-packet behavior, communication load, and \(K\)-dependent runtime. \\
C.3 Boundary/stress & Curious-recipient heatmaps, churn, and anonymous Byzantine fragments, including additional \texttt{FragKrum} supporting plots. \\
\bottomrule
\end{tabularx}
\end{table}

\subsection{Learning Utility and Payload Handling}
\label{supp:exp-utility-payload}

This group supports the learning-utility and payload-handling results in the main paper. The figures are organized from learning behavior to communication workload. Endpoint accuracy appears in the main text, while the supplementary curves show how those endpoints arise over rounds and how payload choices change the amount of meaningful model state exchanged.

Figure~\ref{fig:supp-scale-convergence} shows the full convergence traces behind the scale experiment, for all three datasets and \(N\) up to 100. Increasing \(N\) reduces the amount of local data per node under the fixed partitioning scheme, so local updates become weaker and convergence slows. Fashion-MNIST and CIFAR-10 remain harder than MNIST across all network sizes.

\begin{figure}[H]
    \centering
    \includegraphics[width=\linewidth]{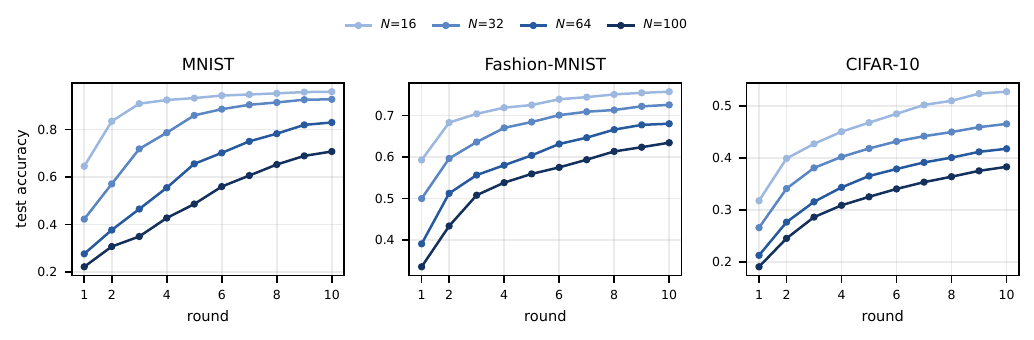}
    \caption{Per-round test accuracy across network sizes \(N\in\{16,32,64,100\}\), for MNIST, Fashion-MNIST, and CIFAR-10 with LeNet-5.}
    \label{fig:supp-scale-convergence}
\end{figure}

Figure~\ref{fig:supp-convergence-systems} compares the per-round trajectories of the three systems at \(N=100\). UnlinkableDFL with the mixnet on and off overlap throughout, so anonymous routing changes the convergence path only marginally. Against the non-anonymous Fedstellar baseline, Fedstellar converges faster on MNIST and Fashion-MNIST, while UnlinkableDFL stays ahead on CIFAR-10 at every round.

\begin{figure}[H]
    \centering
    \includegraphics[width=\linewidth]{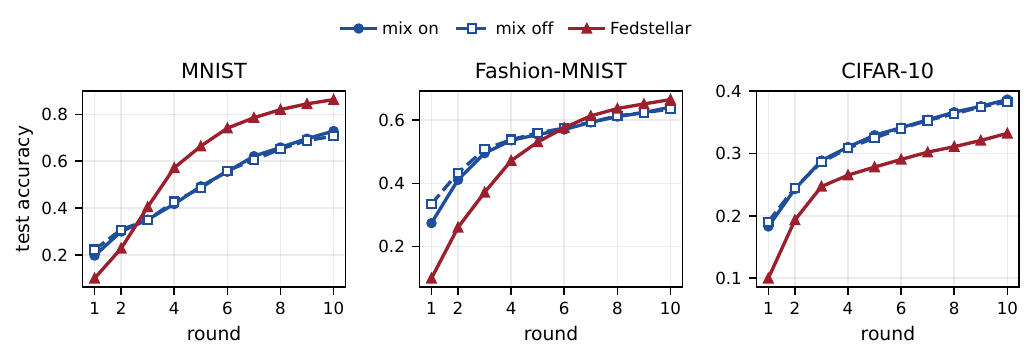}
    \caption{Per-round test accuracy at \(N=100\) for UnlinkableDFL with the mixnet on and off, and for the non-anonymous Fedstellar baseline.}
    \label{fig:supp-convergence-systems}
\end{figure}

Figure~\ref{fig:supp-fragment-volume} explains why payload handling matters before packets enter the mixnet. The fragment-count panel reports the number of Sphinx bodies needed to transmit one model update. The byte-volume panel converts those fragments into per-round network bytes as \(N\) grows. The gap between LeNet-5 and larger architectures shows why quantization, compression, and partial exchange are needed for larger models.

\begin{figure}[H]
    \centering
    \begin{subfigure}[t]{0.48\linewidth}
        \centering
        \includegraphics[width=\linewidth]{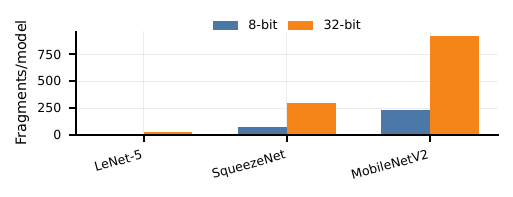}
        \caption{Fragments per model.}
    \end{subfigure}\hfill
    \begin{subfigure}[t]{0.48\linewidth}
        \centering
        \includegraphics[width=\linewidth]{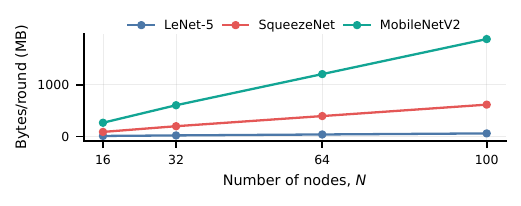}
        \caption{Bytes per round.}
    \end{subfigure}
    \caption{Payload volume by model class and payload encoding.}
    \label{fig:supp-fragment-volume}
\end{figure}

Figure~\ref{fig:supp-partial-rounds} reports the full partial-update trajectories. Sending fewer fragments per round reduces the meaningful exchange workload, but also reduces how much peer information reaches each aggregation step. The effect is modest on MNIST and stronger on Fashion-MNIST, consistent with the harder task and lower endpoint accuracy in the main paper.

\begin{figure}[H]
    \centering
    \begin{subfigure}[t]{0.48\linewidth}
        \centering
        \includegraphics[width=\linewidth]{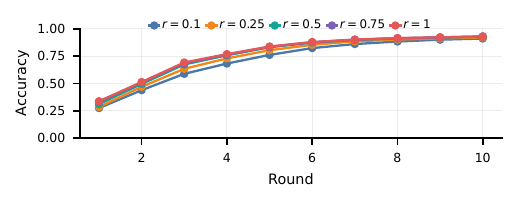}
        \caption{MNIST.}
    \end{subfigure}\hfill
    \begin{subfigure}[t]{0.48\linewidth}
        \centering
        \includegraphics[width=\linewidth]{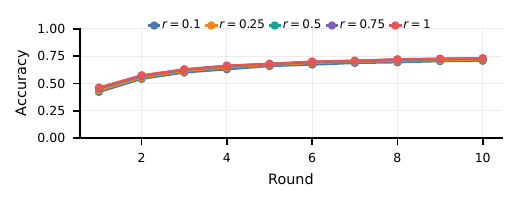}
        \caption{Fashion-MNIST.}
    \end{subfigure}
    \caption{Per-round convergence under partial-update ratios \(r\in\{0.10,0.25,0.50,0.75,1.00\}\).}
    \label{fig:supp-partial-rounds}
\end{figure}

Figure~\ref{fig:supp-heterogeneity} isolates data heterogeneity. Smaller \(\alpha\) produces more skewed local partitions and makes each node's update more distinctive. The heatmap shows the interaction between heterogeneity and partial exchange, while the convergence plot shows the full-update case. These results help interpret both utility loss and the content-based linkability boundary.

\begin{figure}[H]
    \centering
    \begin{subfigure}[t]{0.48\linewidth}
        \centering
        \includegraphics[width=\linewidth]{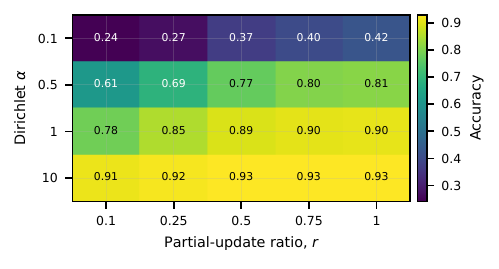}
        \caption{Round-10 accuracy.}
    \end{subfigure}\hfill
    \begin{subfigure}[t]{0.48\linewidth}
        \centering
        \includegraphics[width=\linewidth]{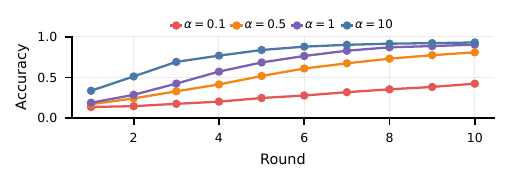}
        \caption{Convergence at \(r=1\).}
    \end{subfigure}
    \caption{Effect of Dirichlet heterogeneity on learning utility.}
    \label{fig:supp-heterogeneity}
\end{figure}

\FloatBarrier
\subsection{Network-Layer Attack, Entropy, and Cost Diagnostics}
\label{supp:exp-attack-entropy}

This group supports the sender-linking attack, entropy discussion, and cost details in the main paper. The entropy figures are diagnostics rather than direct posterior guarantees. They explain how \(K\), \(O\), topology, and degree change the candidate routes and relay shuffles that appear in the theoretical bound, while the cost figures show the overhead of enabling the mixer and increasing path length.

The attack uses the default \(N=32\), degree-4 circulant topology, \(O=100\), and \(\mu=0.1\)~s unless a sweep varies one parameter. The passive adversary observes packet timing, visible transport links, and compromised-relay views, then predicts the sender of a target model-fragment packet. The reported metric is top-1 sender-linking accuracy against the random baseline \(1/|\mathcal{U}_T|\). Ablations remove cover packets, randomized delay, multi-hop forwarding, or the mixnet itself so that attack success can be interpreted against the diagnostics below.

\begin{algorithm}[H]
\small
\caption{Passive network-layer sender-linking adversary}
\label{alg:e2}
\DontPrintSemicolon
\SetKwInOut{Input}{Input}\SetKwInOut{Output}{Output}
\Input{trace \(T\), public \(K,O,\mu\), topology \(\mathcal{G}\), and target delivered to \(d\) at time \(t_{\mathrm{arr}}\) with visible last hop \(p\).}
\Output{sender estimate \(\hat u\) and linking probability \(p_{\mathrm{link}}\).}
\BlankLine
\(\mathcal{U}_T \leftarrow\) nodes within \(K\) hops upstream of \(d\) in \(\mathcal{G}\)\;
\ForEach{candidate \(u \in \mathcal{U}_T\)}{
  \(\ell^{\mathrm{C1}}_u \leftarrow\) predecessor score from last hop and compromised-relay segments\;
  \(\ell^{\mathrm{C2}}_u \leftarrow\) timing score near \(t_{\mathrm{arr}}-h\mu\), \(h=1,\dots,K\)\;
  \(\ell^{\mathrm{C3}}_u \leftarrow\) volume score before \(t_{\mathrm{arr}}\) over uniform expectation\;
  \(\pi(u) \leftarrow \ell^{\mathrm{C1}}_u\,\ell^{\mathrm{C2}}_u\,\ell^{\mathrm{C3}}_u\)\;
}
\(\pi \leftarrow \pi / \sum_{u}\pi(u)\) \tcp*{normalize over \(\mathcal{U}_T\)}
\(\hat u \leftarrow \arg\max_{u}\pi(u)\), \quad \(p_{\mathrm{link}} \leftarrow \max_{u}\pi(u)\)\;
\Return \(\hat u,\ p_{\mathrm{link}}\)
\end{algorithm}

\paragraph{Recorded trace and instrumentation}
The adversary is evaluated on an edge-level trace that is separate from the node-aggregate metrics used elsewhere. Two kinds of records are logged. Adversary-observable records capture what a passive network-layer attacker sees: for every received packet, the overlay edge \(v\!\to\!w\), the arrival time, the packet size, and a hash of the wire bytes. Because Sphinx re-randomizes the bytes at every hop, this hash does not let the attacker follow a packet across hops. A swept fraction \(c\) of compromised relays additionally log the path segment \(v\!\to\!r\!\to\!w\) they can link by peeling one onion layer. Oracle records capture ground truth that rode sealed inside the onion, namely each target packet's true origin and sampled path. They are used only to score the guess and are never exposed as adversary features.

\paragraph{Per-channel features}
For a target \(m\) delivered at \(d\) at time \(t_{\mathrm{arr}}\) with visible last hop \(p\), the candidate set \(\mathcal{U}_T\) is the set of nodes within \(K\) overlay hops upstream of \(d\). Each candidate \(u\) is scored on three channels. The predecessor channel is boosted when \(u=p\) and a single hop cannot hide the origin, or when a compromised relay reports \(u\) as a predecessor inside \([t_{\mathrm{arr}}-K\mu,\,t_{\mathrm{arr}}]\). The timing channel compares \(u\)'s outbound packet count near \(t_{\mathrm{arr}}-h\mu\), for \(h=1,\dots,K\), against its baseline emission rate. The volume channel compares \(u\)'s outbound count in \([t_{\mathrm{arr}}-K\mu,\,t_{\mathrm{arr}}]\) with its uniform expectation. The three log-likelihoods are summed under a conditional-independence approximation, the posterior is normalized over \(\mathcal{U}_T\), and the maximum-a-posteriori candidate is the guess.

Algorithm~\ref{alg:e2} operationalizes this adversary for the experiments. It is not part of the protocol. The first step restricts inference to senders that are topologically consistent with the observed delivery. The loop then assigns every candidate one positive score per channel: C1 captures direct predecessor evidence, C2 captures whether the candidate emitted traffic at a compatible time, and C3 captures whether the candidate's recent traffic volume is unusually high. The product is a posterior proxy over the feasible candidates. If no channel provides evidence, the scores remain close across candidates and the guess approaches the random baseline.

The implementation evaluates this procedure independently for each target packet and then reports top-1 accuracy over all targets in a run. Ground-truth sender labels are used only at this scoring stage. The normalized value \(p_{\mathrm{link}}\) is retained as a diagnostic for how concentrated the adversary's posterior is, while the main paper reports whether the selected sender \(\hat u\) matches the sealed oracle origin.

Figure~\ref{fig:supp-entropies} separates the relay and path views. Relay entropy reflects uncertainty inside one relay outbox, while path entropy counts feasible shuffled paths under the uniform route-counting setup. Increasing \(K\) has the strongest effect on path entropy, while increasing \(O\) mainly increases per-relay shuffle uncertainty.

\begin{figure}[H]
    \centering
    \begin{subfigure}[t]{0.48\linewidth}
        \centering
        \includegraphics[width=\linewidth]{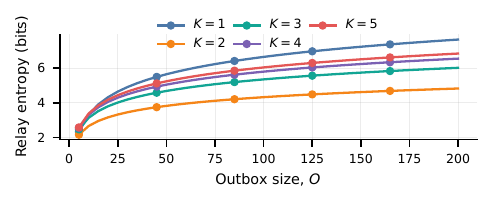}
        \caption{Relay entropy.}
    \end{subfigure}\hfill
    \begin{subfigure}[t]{0.48\linewidth}
        \centering
        \includegraphics[width=\linewidth]{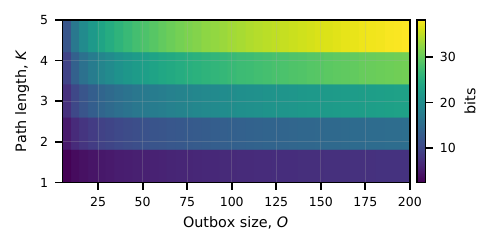}
        \caption{Path entropy.}
    \end{subfigure}
    \caption{Entropy diagnostics used by the privacy-cost analysis.}
    \label{fig:supp-entropies}
\end{figure}

Figure~\ref{fig:supp-topology} shows the topology side of the same trade-off. A larger reachable set gives the adversary more candidate peers to distinguish, but higher degree also increases neighbor state and connection-management work. The two-hop reach plot explains why degree and graph family matter even when \(K\) and \(O\) are fixed.

\begin{figure}[H]
    \centering
    \begin{subfigure}[t]{0.45\linewidth}
        \centering
        \includegraphics[width=\linewidth]{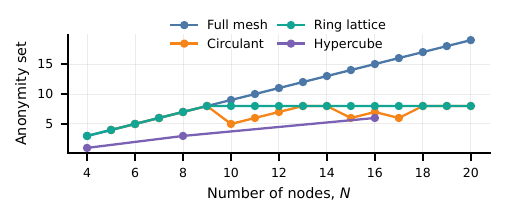}
        \caption{Topology family.}
    \end{subfigure}
    \vspace{0.3em}
    \begin{subfigure}[t]{0.9\linewidth}
        \centering
        \includegraphics[width=\linewidth]{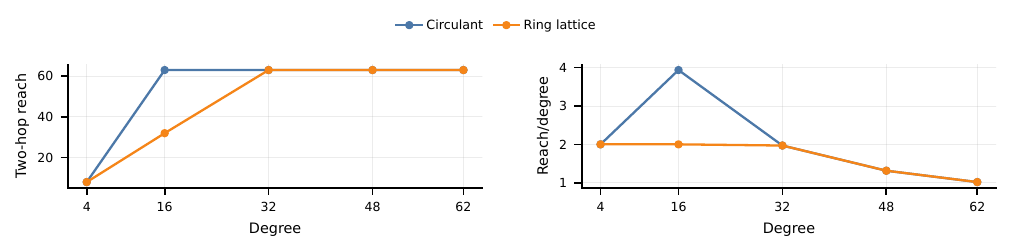}
        \caption{Graph degree.}
    \end{subfigure}
    \caption{Topology-dependent candidate-set diagnostics at \(K=2\).}
    \label{fig:supp-topology}
\end{figure}

\paragraph{Cost diagnostics.}
The following figures support the unlinkability--cost trade-off in the main paper. They focus on the cost of enabling the mixer, the cover-packet behavior induced by the release schedule, and the runtime impact of larger path lengths.

Figure~\ref{fig:supp-mixing-cost} compares mixed and non-mixed execution across graph degree. The mix-on condition enables onion forwarding, buffering, cover packets, and randomized release. The mix-off condition keeps the same learning workload but removes anonymous transport, so the comparison isolates the fixed price of enabling the mixer.

\begin{figure}[H]
    \centering
    \begin{subfigure}[t]{0.48\linewidth}
        \centering
        \includegraphics[width=\linewidth]{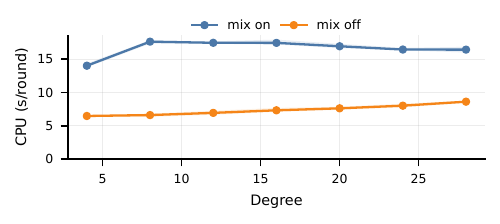}
        \caption{CPU time.}
    \end{subfigure}\hfill
    \begin{subfigure}[t]{0.48\linewidth}
        \centering
        \includegraphics[width=\linewidth]{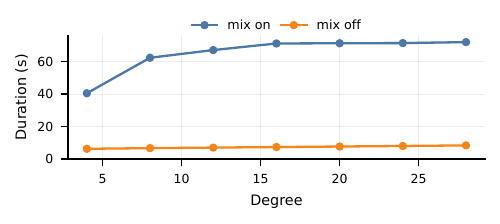}
        \caption{Round duration.}
    \end{subfigure}
    \caption{Mix-on and mix-off cost across graph degree at \(N=32\).}
    \label{fig:supp-mixing-cost}
\end{figure}

Figure~\ref{fig:supp-cover-mu} reports how cover-packet padding reacts to the delay mean \(\mu\). A smaller \(\mu\) creates more release opportunities during an active round, while the number of meaningful fragments is fixed by the model payload and exchange ratio. Empty release opportunities are filled with cover packets, so the cover fraction and per-node cover count rise as \(\mu\) decreases.

\begin{figure}[H]
    \centering
    \begin{subfigure}[t]{0.48\linewidth}
        \centering
        \includegraphics[width=\linewidth]{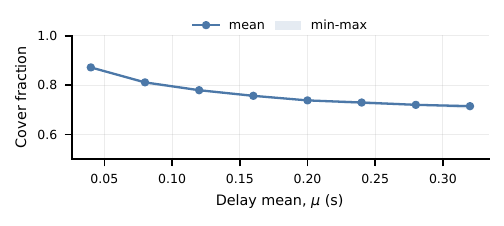}
        \caption{Cover fraction.}
    \end{subfigure}\hfill
    \begin{subfigure}[t]{0.48\linewidth}
        \centering
        \includegraphics[width=\linewidth]{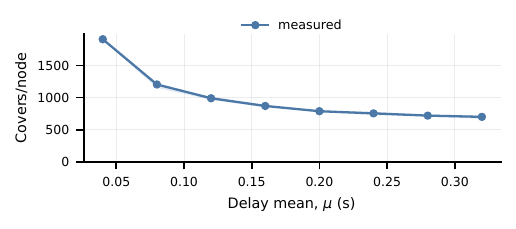}
        \caption{Cover packets per node.}
    \end{subfigure}
    \caption{Cover-packet padding versus delay mean \(\mu\) at \(N=32\).}
    \label{fig:supp-cover-mu}
\end{figure}

Figure~\ref{fig:supp-bandwidth} reports the absolute communication cost behind these fractions. Cover and relay packets are forwarded over multi-hop paths like real fragments, so total network bytes rise with both network size and path length. The per-node message load grows with \(N\) and with the partial-update ratio \(r\). Total bandwidth is therefore a main scalability cost of the mixer.

\begin{figure}[H]
    \centering
    \begin{subfigure}[t]{0.48\linewidth}
        \centering
        \includegraphics[width=\linewidth]{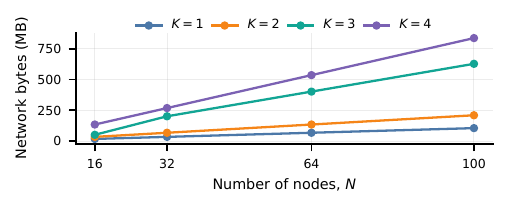}
        \caption{Total network bytes.}
    \end{subfigure}\hfill
    \begin{subfigure}[t]{0.48\linewidth}
        \centering
        \includegraphics[width=\linewidth]{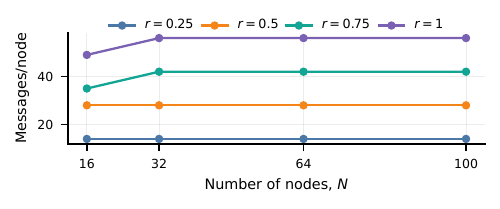}
        \caption{Per-node message load.}
    \end{subfigure}
    \caption{Communication cost including cover packets.}
    \label{fig:supp-bandwidth}
\end{figure}

Figure~\ref{fig:supp-k-cost} shows why large \(K\) values are expensive. Longer paths increase RTT, CPU time, and outstanding fragments because packets spend more time in the overlay. This supports the main paper's operating-region statement that moderate \(K\) is the practical choice for rounds that can tolerate tens of seconds.

\begin{figure}[H]
    \centering
    \begin{subfigure}[t]{0.48\linewidth}
        \centering
        \includegraphics[width=\linewidth]{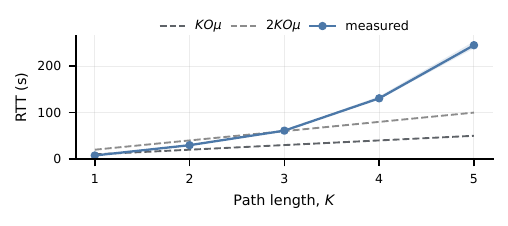}
        \caption{RTT.}
    \end{subfigure}\hfill
    \begin{subfigure}[t]{0.48\linewidth}
        \centering
        \includegraphics[width=\linewidth]{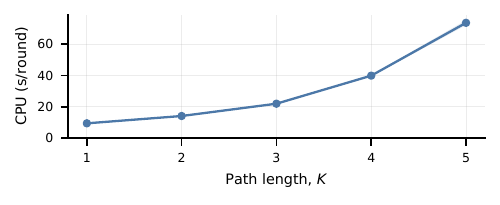}
        \caption{CPU time.}
    \end{subfigure}

    \vspace{0.5em}
    \begin{subfigure}[t]{0.48\linewidth}
        \centering
        \includegraphics[width=\linewidth]{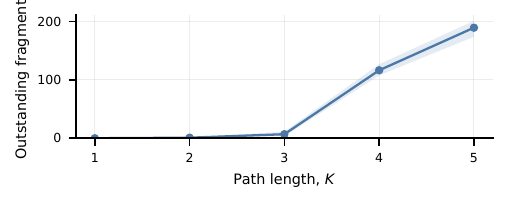}
        \caption{Outstanding fragments.}
    \end{subfigure}\hfill
    \begin{subfigure}[t]{0.48\linewidth}
        \centering
        \includegraphics[width=\linewidth]{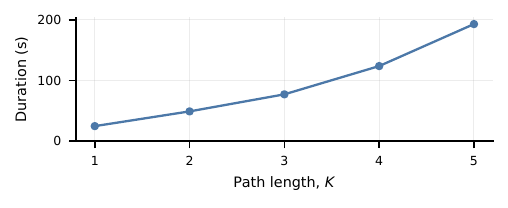}
        \caption{Round duration.}
    \end{subfigure}
    \caption{Cost response as the maximum path length \(K\) grows.}
    \label{fig:supp-k-cost}
\end{figure}

\FloatBarrier
\subsection{Boundary and Stress Tests}
\label{supp:exp-boundary-stress}

This group collects experiments that clarify the framework's boundary and stress behavior. The curious-recipient experiment is outside the network-layer threat model, but it documents the payload-layer fingerprints that remain after fragments are decrypted. Churn and Byzantine experiments test whether the prototype still exchanges useful fragments under changing or corrupted peers.

Figure~\ref{fig:supp-content-attack} expands the content-based boundary test from the main paper. The recipient receives a first-round reference update, then compares decrypted later fragments by cosine distance. Smaller \(\alpha\) and larger \(r\) produce stronger content fingerprints. Identity linking is harder than bucket tracking because it must assign a same-sender bucket to a concrete participant.

\begin{figure}[H]
    \centering
    \begin{subfigure}[t]{0.48\linewidth}
        \centering
        \includegraphics[width=\linewidth]{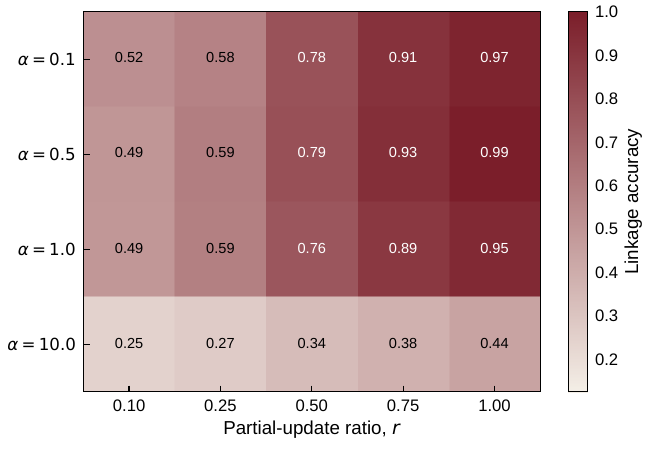}
        \caption{Identity linking.}
    \end{subfigure}\hfill
    \begin{subfigure}[t]{0.48\linewidth}
        \centering
        \includegraphics[width=\linewidth]{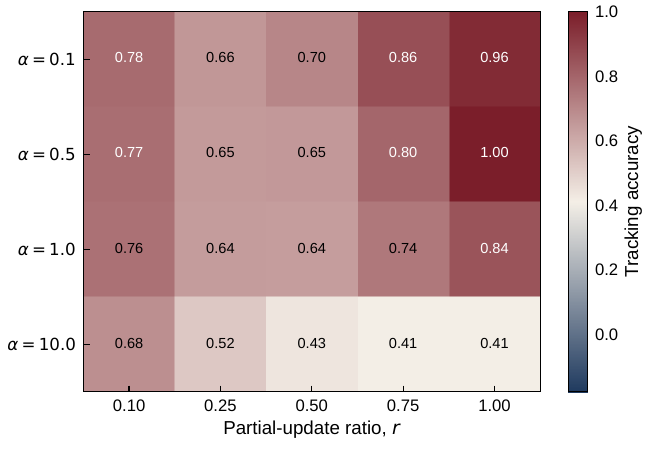}
        \caption{Bucket tracking.}
    \end{subfigure}
    \caption{Curious-recipient content attack over data heterogeneity \(\alpha\) and partial-update ratio \(r\).}
    \label{fig:supp-content-attack}
\end{figure}

Figure~\ref{fig:supp-churn} combines the two churn scenarios. In the peer-exit case, one peer leaves after round 5. Remaining peers are grouped by exposure to the departed peer: direct transport neighbors, overlay peers that lose a two-hop route, and unaffected peers outside the departed peer's two-hop exchange region. In the late-join case, one peer is absent from rounds 1--2 and starts participating in round 3. The curves show that churn affects fragment coverage locally and temporarily rather than requiring a global restart. The topology class definitions are reported in the text instead of as a separate network diagram to keep the figure readable.

\begin{figure}[H]
    \centering
    \begin{subfigure}[t]{0.45\linewidth}
        \centering
        \includegraphics[width=\linewidth]{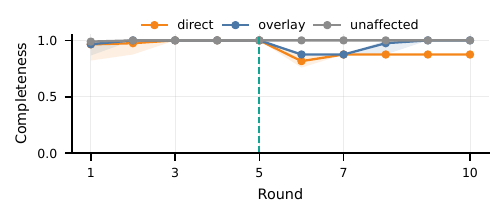}
        \caption{Peer-exit completeness.}
    \end{subfigure}

    \vspace{0.5em}
    \begin{subfigure}[t]{0.45\linewidth}
        \centering
        \includegraphics[width=\linewidth]{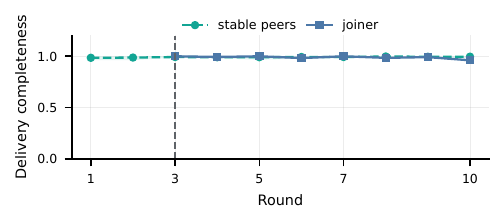}
        \caption{Joiner fragment intake.}
    \end{subfigure}\hfill
    \begin{subfigure}[t]{0.45\linewidth}
        \centering
        \includegraphics[width=\linewidth]{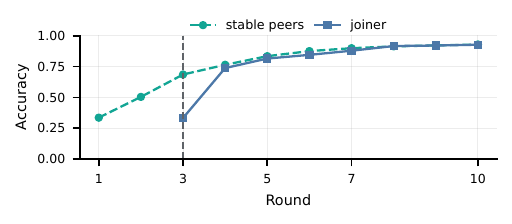}
        \caption{Joiner accuracy.}
    \end{subfigure}
    \caption{Peer exit and late join at \(N=32\), degree four, and \(K=2\).}
    \label{fig:supp-churn}
\end{figure}

Figure~\ref{fig:supp-krum-attacks} extends the two Byzantine attacks shown in the main paper with a sign-flip stressor. Malicious peers produce corrupted model values, which are fragmented and delivered without sender labels. \texttt{FragFedAvg} is sensitive to high-magnitude corrupted values because the mean can be dominated by outliers. \texttt{FragKrum} applies a range-wise nearest-neighbor score and remains effective when the per-range condition has enough candidates.

\begin{figure}[H]
    \centering
    \includegraphics[width=0.7\linewidth]{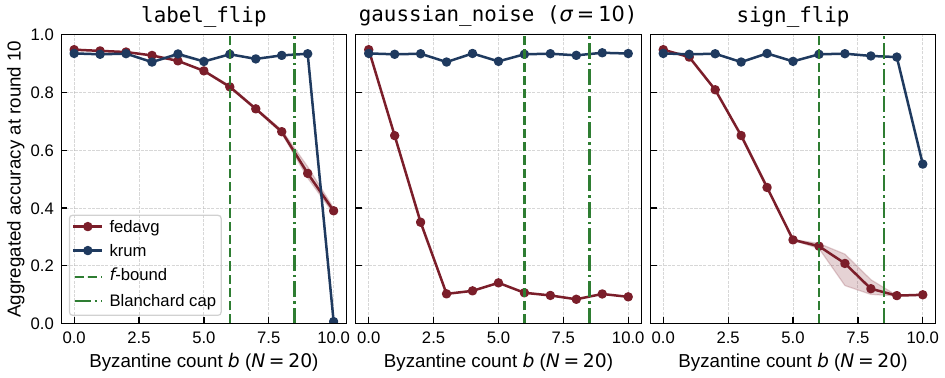}
    \caption{\texttt{FragFedAvg} and \texttt{FragKrum} under label-flip, Gaussian-noise, and sign-flip attacks at \(N=20\).}
    \label{fig:supp-krum-attacks}
\end{figure}
\FloatBarrier

\end{document}